\documentclass[12pt,english]{article}
\usepackage{mathptmx}

\usepackage[T1]{fontenc}
\usepackage[a4paper]{geometry}
\usepackage{amsthm}
\usepackage{amsmath}
\usepackage{amssymb}
\usepackage{graphicx}
\usepackage{setspace}
\usepackage[authoryear]{natbib}
\makeatletter
\numberwithin{table}{section}
\numberwithin{figure}{section}
\usepackage{enumitem}		
\theoremstyle{plain}
\newtheorem{thm}{\protect\theoremname}[section]
\theoremstyle{definition}
\newtheorem{defn}[thm]{\protect\definitionname}
\ifx\proof\undefined
\newenvironment{proof}[1][\protect\proofname]{\par
\normalfont\topsep6\p@\@plus6\p@\relax
\trivlist
\itemindent\parindent
\item[\hskip\labelsep
\scshape
#1]\ignorespaces
}{%
\endtrivlist\@endpefalse
}
\providecommand{\proofname}{Proof}
\fi
\theoremstyle{plain}
\newtheorem{lem}[thm]{\protect\lemmaname}
\theoremstyle{plain}
\newtheorem{cor}[thm]{\protect\corollaryname}
\theoremstyle{definition}
\newtheorem{example}[thm]{\protect\examplename}

\usepackage{amsfonts,amssymb,amsthm,amsbsy,amssymb,amsmath,mathrsfs}
\usepackage{txfonts,graphicx,mathrsfs,rotating,array,sectsty}
\usepackage[symbol]{footmisc}
\usepackage{natbib}
\usepackage{fancyref}
\usepackage[linkcolor=blue,citecolor=blue, colorlinks]{hyperref}
\usepackage{bm}

\newcommand{\q}{\quad}
\newcommand{\qq}{\qquad}
\newcommand{\mcl}[1]{\mathcal{#1}}

\newcommand{\mbb}[1]{\mathbb{#1}}
\newcommand{\pr}{\hbox{\sf P}}

\newcommand{\ep}{\hbox{\sf E}}

\newcommand{\tto}{\rightarrow}

\newcommand{\dt}{\mathrm{d}}

\setcitestyle{round}

\numberwithin{equation}{section}

\allowdisplaybreaks

\sectionfont{\rmfamily\mdseries\large\bf}
\subsectionfont{\rmfamily\mdseries\normalsize\bf}
\subsubsectionfont{\rmfamily\mdseries\normalsize\it}

\makeatother

\usepackage{babel}
\providecommand{\corollaryname}{Corollary}
\providecommand{\definitionname}{Definition}
\providecommand{\examplename}{Example}
\providecommand{\lemmaname}{Lemma}
\providecommand{\theoremname}{Theorem}

\begin{document}

\title{\textbf{\Large{On the Dividend Strategies with Non-Exponential Discounting}}}

\author{Qian Zhao \thanks{School of Finance and Statistics, East China Normal University, Shanghai, 200241, China. E-mail: qzhao31@gmail.com},
\q Jiaqin Wei\thanks{Corresponding author. Department of Applied Finance and Actuarial Studies, Faculty of Business and Economics, Macquarie University, Sydney, NSW 2109, Australia. E-mail: jiaqinwei@gmail.com},
\q Rongming Wang\thanks{School of Finance and Statistics, and Research Center of International Finance and Risk Management, East China Normal University, Shanghai, 200241, China. E-mail: rmwang@stat.ecnu.edu.cn}}

\date{\normalsize\it\today}
\maketitle
\begin{abstract}
In this paper, we study the dividend strategies for a shareholder
with non-constant discount rate in a diffusion risk model. We assume
that the dividends can only be paid at a bounded rate and restrict
ourselves to Markov strategies. This is a time inconsistent control
problem. The equilibrium HJB-equation is given and the verification
theorem is proved for a general discount function. Considering a mixture
of exponential discount functions and a pseudo-exponential discount
function, we get equilibrium dividend strategies and the corresponding
equilibrium value functions by solving the equilibrium HJB-equations.

\end{abstract}

\section{Introduction }

Since it was proposed by \citet{df57}, the optimization of dividend
strategy has been investigated by many researchers under various risk
models. This problem is usually phrased as the management\textquoteright{}s
problem of determining the optimal timing and the size of dividend
payments in the presence of bankruptcy risk. For more literature on
this problem, we refer the reader to a recent survey paper by \citet{a09}. 

In the very rich literature, a common assumption is that the discount
rate is constant over time so the discount function is exponential.
However, some empirical studies of human behavior suggest that the
assumption of constant discount rate is unrealistic, see, e.g., \citet{t81},
\citet{a92} and \citet{lp92}. Indeed, there is experimental evidence
that people are impatient about choices in the short term but are
more patient when choosing between long-term alternatives. More precisely,
events in the near future tend to be discounted at a higher rate than
events that occur in the long run. Considering such effect, individual
behavior is best described by the hyperbolic discounting (see \citet{pp68}),
which has been extensively studied in the areas of microeconomics,
macroeconomics, and behavioral finance, such as \citet{l97} and \citet{b99}
among others.

However, difficulties arise when we try to solve an optimal control
problem with a non-constant discount rate by the standard dynamic
programming approach. In fact, the standard optimal control techniques
give rise to time inconsistent strategies, i.e, a strategy that is
optimal for the initial time may be not optimal later. This is the
so-called time inconsistent control problem and the classical dynamic
programming principle does no longer hold. \citet{s55} studies the
time inconsistent problem within a game theoretic framework by using
Nash equilibrium points. They seek the equilibrium policy as the solution
of a subgame-perfect equilibrium where the players are the agent and
her future selves. 

Recently, there is an increasing attention in the time inconsistent
control problem due to the practical applications in economics and
finance. A modified HJB equation is derived in \citet{mn10} which
solves an optimal consumption and investment problem with the non-constant
discount rate for both naive and sophisticated agents. A similar problem
is also considered by another approach in \citet{el06} and \citet{ep08},
which provide the precise definition of the equilibrium concept in
continuous time for the first time. They characterize the equilibrium
policies through the solutions of a flow of BSDEs, and they show,
for a special form of the discount factor, that this BSDE reduces
to a system of two ODEs which has a solution. Considering the hyperbolic
discounting, \citet{emp12} studies the portfolio management problem
for an investor who is allowed to consume and take out life insurance,
and they characterize the equilibrium strategy by an integral equation.
Following this definition of the equilibrium strategy, \citet{bm10}
studied the time-inconsistent control problem in a general Markov
framework, and derived the equilibrium HJB-equation together with
the verification theorem. \citet{bmz12} studied the Markowitz's problem
with state-dependent risk aversion by utilizing the equilibrium HJB-equation
obtained in \citet{bm10}.

In this paper, we study the dividend strategies for the shareholders
with non-constant discount rate in a diffusion risk model. We assume
that the dividends can only be paid at a bounded rate and restrict
ourselves to Markov strategies. We use the equilibrium HJB-equation
to solve this problem. In contrast to the papers mentioned above which
consider a fixed time horizon or an infinite time horizon, in the
dividend problem the ruin risk should be taken into account and the
time horizon is a random variable (the time of ruin). Thus, the equilibrium
HJB-equation given in this paper looks different with the one obtained
in \citet{bm10}. We first give the equilibrium HJB-equation which
is motivated by \citet{y12} and the verification theorem for a general
discount function. Then we solve the equilibrium HJB-equation for
two special non-exponential discount functions: a mixture of exponential
discount function and a pseudo-exponential discount function. For
more details about these discount functions, we refer the reader to
\citet{el06} and \citet{ep08}. Under the mixture of exponential
discount function, our results show that if the bound of the dividend
rate is small enough then the equilibrium strategy is to always pay
the maximal dividend rate; otherwise, the equilibrium strategy is
to pay the maximal dividend rate when the surplus is above a barrier
and pay nothing when the surplus is below the barrier. Given some
conditions, the results are similar under the pseudo-exponential discount
function. These features of the equilibrium dividend strategies are
similar to the optimal strategies obtained in \citet{at97} which
considers the exponential discounting in the diffusion risk model.

The remainder of this paper is organized as follows. The dividend
problem and the definition of an equilibrium strategy are given in
Section 2. The equilibrium HJB-equation and a verification theorem
are presented in Section 3. In Section 4, we study two cases with
a mixture of exponential discount functions and a pseudo-exponential
discount function.

\section{The model}

In the case of no control, the surplus process is assumed to follow
\[
\dt X_{t}=\mu\dt t+\sigma\dt W_{t},\qq t\geq0,
\]
where $\mu,\sigma$ are positive constants and $\{W_{t}\}_{t\geq0}$
is a one-dimensional standard Brownian motion on a filtered probability
space $\left(\Omega,\mathcal{F},\left\{ \mathcal{F}_{t}\right\} _{t\geq0},\pr\right)$
satisfying the usual hypotheses. The filtration $\left\{ \mathcal{F}_{t}\right\} _{t\geq0}$
is completed and generated by $\left\{ W_{t}\right\} _{t\geq0}$. 

A dividend strategy is described by a stochastic process $\left\{ l_{t}\right\} _{t\geq0}$.
Here, $l_{t}\geq0$ is the rate of dividend payout at time $t$ which
is assumed to be bounded by a constant $M>0$. We restrict ourselves
to the feedback control strategies (Markov strategies), i.e. at time
$t$, the control $l_{t}$ is given by 
\[
l_{t}=\pi(t,x),
\]
where $x$ is the surplus level at time $t$ and the control law $\pi:[0,\infty)\times[0,\infty)\tto[0,M]$
is a Borel measurable function. 

When applying the control law $\pi$, we denote by $\{X_{t}^{\pi}\}_{t\geq0}$
the controlled risk process. Considering the controlled system starting
from the initial time $t\in[0,\infty)$, $\{X_{s}^{\pi}\}$ evolves
according to
\begin{align}
\begin{cases}
\dt X_{s}^{\pi} & =\mu\dt s+\sigma\dt W_{s}-\pi(s,X_{s}^{\pi})\dt s,\q s\geq t,\\
X_{t}^{\pi} & =x.
\end{cases}\label{eq:2-0}
\end{align}

Let 
\[
\tau_{t}^{\pi}:=\inf\left\{ s\geq t:X_{s}^{\pi}\leq0\right\} 
\]
be the time of ruin under the control law $\pi$. Without loss of
generality, we assume that $X_{s}^{\pi}\equiv0$ for $s\geq\tau_{t}^{\pi}$. 

Let $h:[0,\infty)\rightarrow[0,\infty)$ be a discount function which
satisfies $h(0)=1$, $h(s)\geq0$ and $\int_{0}^{\infty}h(t)\dt t<\infty$.
Furthermore, $h$ is assumed to be continuously differentiable on
$[0,\infty)$ and $h'(x)\leq0$.
\begin{defn}
\label{def:2-1}A control law $\pi$ is said to be admissible if it
satisfies: $0\leq\pi(t,x)\leq M$ for all $(t,x)\in[0,\infty)\times[0,\infty)$,
$\pi(t,0)\equiv0$ for all $t\in[0,\infty)$. We denote by $\Pi$
the set of all admissible control laws. 
\end{defn}
For a given admissible control law $\pi$ and an initial state $(t,x)\in[0,\infty)\times[0,\infty)$,
we define the return function $V^{\pi}$ by 
\[
V^{\pi}(t,x)=\ep_{t,x}\left[\int_{t}^{\tau_{t}^{\pi}}h(s-t)\pi(s,X_{s}^{\pi})\dt s\right],
\]
where $\ep_{t,x}[\cdot]$ is the expectation conditioned on the event
$\{X_{t}^{\pi}=x\}$. Note that for any admissible strategy $\pi\in\Pi$,
we have
\begin{equation}
\ep_{t,x}\left[\int_{t}^{\tau_{t}^{\pi}}\left|h(s-t)\pi(s,X_{s}^{\pi})\right|\dt s\right]\leq M\int_{0}^{\infty}h(t)\dt t<\infty,\q\forall(t,x)\in[0,\infty)\times[0,\infty),\label{eq:a-0}
\end{equation}
which means the performance functions $V^{\pi}(t,x)$ are well-defined
for all admissible strategies.

In classical risk theory, the optimal dividend strategy, denoted by
$\pi^{*}$, is an admissible strategy such that 
\[
V^{\pi^{*}}(t,x)=\sup_{\pi\in\Pi}V^{\pi}(t,x).
\]
However, in our settings, this optimization problem is time-inconsistent
in the sense that the Bellman optimality principle fails. 

Similar to \citet{ep08} and \citet{bm10}, we view the entire problem
as a non-cooperative game and look for Nash equilibria for the game.
More specifically, we consider a game with one player for each time
$t$, where player $t$ can be regarded as the future incarnation
of the decision maker at time $t$. Given state $(t,x)$, player $t$
will choose a control action $\pi(t,x)$, and she/he wants to maximize
the functional $V^{\pi}(t,x)$. In the continuous-time model, \citet{el06}
and \citet{ep08} give the precise definition of this equilibrium
strategy for the first time. Intuitively, equilibrium strategies are
the strategies such that, given that they will be implemented in the
future, it is optimal to implement them right now. 
\begin{defn}
\label{def:2-2}Choose a control law $\hat{\pi}\in\Pi$, a fixed $l\in[0,M]$
and a fixed real number $\epsilon>0$. For any fixed initial point
$(t,x)\in[0,\infty)\times[0,\infty),$ we define the control law $\pi^{\epsilon}$
by 
\[
\pi^{\epsilon}(s,y)=\begin{cases}
0, & \mathrm{for}\ s\in[t,\infty),\ y=0;\\
l, & \mathrm{for}\ s\in[t,t+\epsilon],\ y\in(0,\infty);\\
\hat{\pi}(s,y), & \mathrm{for}\ s\in[t+\epsilon,\infty),\ y\in(0,\infty).
\end{cases}
\]
If 
\[
\liminf_{\epsilon\rightarrow0}\frac{V^{\hat{\pi}}(t,x)-V^{\pi^{\epsilon}}(t,x)}{\epsilon}\geq0,
\]
for all $l\in[0,M],$ we say that $\hat{\pi}$ is an equilibrium control
law. And the equilibrium value function $V$ is defined by 
\begin{equation}
V(t,x)=V^{\hat{\pi}}(t,x).\label{eq:2-2}
\end{equation}

\end{defn}
In the following section, we will first give the equilibrium HJB-equation
for the equilibrium value function $V$, and then prove a verification
theorem.

\section{The equilibrium Hamilton-Jacobi-Bellman Equation}

In this section, we consider the objective function having the form
\begin{equation}
V^{\pi}(t,x)=\ep_{t,x}\left[\int_{t}^{\tau_{t}^{\pi}}C\left(t,s,\pi(s,X_{s}^{\pi})\right)\dt s\right],\label{eq:3-1}
\end{equation}
where $C\left(t,s,\pi(s,X_{s}^{\pi})\right)=h(s-t)\pi(s,X_{s}^{\pi}),$
for $s\geq t$.

For all $\pi\in\Pi$ and any real valued function $f(t,x)\in C^{1,2}\left([0,\infty)\times(0,\infty)\right)$,
which means that the partial derivatives $\frac{\partial f}{\partial t},\ \frac{\partial f}{\partial x},\ \frac{\partial^{2}f}{\partial x^{2}}$
exist and are continuous on $[0,\infty)\times(0,\infty)$, we define
the infinitesimal generator $\mcl L^{\pi}$ by
\[
\mcl L^{\pi}f(t,x)=\frac{\partial f}{\partial t}(t,x)+\left(\mu-\pi(t,x)\right)\frac{\partial f}{\partial x}(t,x)+\frac{1}{2}\sigma^{2}\frac{\partial^{2}f}{\partial x^{2}}(t,x).
\]

Let $\mcl D[0,\infty):=\left\{ (s,t)\;\mid\;0\leq s\leq t<\infty\right\} $.
The following equilibrium HJB-equation is motivated by Equation (4.77)
of \citet{y12} and the proof of Theorem \ref{thm:3-2}.
\begin{defn}
\label{def:3-1}For a smooth function $c(s,t,x)$ defined on $\mcl D[0,\infty)\times[0,\infty)$,
the equilibrium HJB-equation is given by
\begin{equation}
\begin{cases}
\frac{\partial c}{\partial t}(s,t,x)+H\left(s,t,\phi\left(t,t,\frac{\partial c}{\partial x}(t,t,x),\frac{\partial^{2}c}{\partial x^{2}}(t,t,x)\right),\frac{\partial c}{\partial x}(s,t,x),\frac{\partial^{2}c}{\partial x^{2}}(s,t,x)\right)=0,\\
\qq\qq\qq\qq\qq\qq\qq\qq\qq\qq\forall(s,t,x)\in\mcl D[0,\infty)\times(0,\infty),\\
c(s,t,0)=0,\q\forall(s,t)\in\mcl D[0,\infty),
\end{cases}\label{eq:3-2}
\end{equation}
where
\begin{equation}
\begin{cases}
H(s,t,l,p,P)=\frac{1}{2}\sigma^{2}P+(\mu-l)p+C(s,t,l),\\
\phi(s,t,p,P)=\arg\max H(s,t,\cdot,p,P),
\end{cases}\label{eq:3-3}
\end{equation}
for $(s,t,l,p,P)\in\mcl D[0,\infty)\times[0,M]\times\mbb R^{2}.$ 
\end{defn}
Since the equilibrium HJB-equation given in Definition \ref{def:3-1}
is informal, we are now giving a strict verification theorem. 
\begin{thm}
\label{thm:3-2}(Verification Theorem) Assume that there exists a
bounded function $c(s,t,x)$, which is smooth enough, solves the equilibrium
HJB-equation in Definition \ref{def:3-1}. Let 
\begin{equation}
\hat{\pi}(t,x):=\phi\left(t,t,\frac{\partial c}{\partial x}(t,t,x),\frac{\partial^{2}c}{\partial x^{2}}(t,t,x)\right)\label{eq:3-5}
\end{equation}
and 
\begin{equation}
V(t,x):=c(t,t,x).\label{eq:3-5-1}
\end{equation}
If for any $(s,t,x)\in\mcl D[0,\infty)\times[0,\infty)$ it holds
that
\begin{equation}
\lim_{n\tto\infty}c(s,\tau_{n},X_{\tau_{n}}^{\hat{\pi}})=0,\q a.s.,\label{eq:3-4}
\end{equation}
where $\tau_{n}=n\wedge\tau_{t}^{\hat{\pi}},n\geq t,\; n=1,2,\cdots,$
and $X^{\hat{\pi}}$ is the unique solution to the SDE (\ref{eq:2-0})
with $\pi$ replaced by $\hat{\pi}$ and initial state $(t,x)$, then
$\hat{\pi}$ given by (\ref{eq:3-5}) is an equilibrium control law,
and $V$ given by (\ref{eq:3-5-1}) is the corresponding equilibrium
value function. \end{thm}
\begin{proof}
We give the proof in two steps: 1. We show that $V$ is the value
function corresponding to $\hat{\pi}$, i.e., $V(t,x)=V^{\hat{\pi}}(t,x)$;
2. We prove that $\hat{\pi}$ is indeed the equilibrium control law
which is defined by Definition \ref{def:2-2}.

Step 1. 

With (\ref{eq:3-5}), we rewrite (\ref{eq:3-2}) as 
\begin{equation}
\begin{cases}
\mcl L^{\hat{\pi}}c(s,t,x)+C(s,t,\hat{\pi}(t,x))=0, & (s,t,x)\in\mcl D[0,\infty)\times(0,\infty),\\
c(s,t,0)=0,\q\forall(s,t)\in\mcl D[0,\infty),
\end{cases}\label{eq:3-6}
\end{equation}
where the operator $\mcl L^{\hat{\pi}}$ applies to the function $c(s,\cdot,\cdot).$

By (\ref{eq:3-6}), applying Dynkin's formula to the function $c(s,\cdot,\cdot)$
yields that 
\begin{eqnarray*}
c(s,t,x) & = & \ep_{t,x}\left[c\left(s,\tau_{n},X_{\tau_{n}}^{\hat{\pi}}\right)\right]-\ep_{t,x}\left[\int_{t}^{\tau_{n}}\mcl L^{\hat{\pi}}c\left(s,z,X_{z}^{\hat{\pi}}\right)\dt z\right]\\
 & = & \ep_{t,x}\left[c\left(s,\tau_{n},X_{\tau_{n}}^{\hat{\pi}}\right)\right]+\ep_{t,x}\left[\int_{t}^{\tau_{n}}C(s,z,\hat{\pi}(z,X_{z}^{\hat{\pi}}))\dt z\right].
\end{eqnarray*}
Recalling Definition \ref{def:2-1} of admissible strategies (see
also (\ref{eq:a-0})), for given $s\leq t$, we have 
\[
\ep_{t,x}\left[\int_{t}^{\tau_{t}^{\hat{\pi}}}\left|C\left(s,z,\hat{\pi}(z,X_{z}^{\hat{\pi}})\right)\right|\dt z\right]<\infty,\q\forall(t,x)\in[0,\infty)\times[0,\infty).
\]
Since $c(\cdot,\cdot,\cdot)$ is bounded, by (\ref{eq:3-4}), letting
$n\tto\infty$ and applying dominated convergence theorem yield that
\begin{equation}
c(s,t,x)=\ep_{t,x}\left[\int_{t}^{\tau_{t}^{\hat{\pi}}}h(z-s)\hat{\pi}(z,X_{z}^{\hat{\pi}})\dt z\right],\q(s,t,x)\in\mcl D[0,\infty)\times[0,\infty).\label{eq:4-1}
\end{equation}
Thus, we have 
\[
V(t,x):=c(t,t,x)=V^{\hat{\pi}}(t,x).
\]

Step 2. For a given $l\in[0,M]$, and a fixed real number $\epsilon>0$,
we define $\pi^{\epsilon}$ by Definition \ref{def:2-2}. For simplicity,
we denote by $X^{\epsilon}$ the path under the control law $\pi^{\epsilon}$.
Without loss of generality, we consider the case where $\epsilon$
is sufficient small such that $t+\epsilon<\tau_{t}^{\pi^{\epsilon}}\wedge\tau_{t}^{\hat{\pi}}$.
By the definition of $V^{\pi}$, we obtain 
\begin{eqnarray}
V^{\hat{\pi}}(t,x)-V^{\pi^{\epsilon}}(t,x) & = & \ep_{t,x}\left[\int_{t}^{\tau_{t}^{\hat{\pi}}}C\left(t,s,\hat{\pi}\left(s,X_{s}^{\hat{\pi}}\right)\right)\dt s-\int_{t}^{\tau_{t}^{\pi^{\epsilon}}}C\left(t,s,\pi^{\epsilon}\left(s,X_{s}^{\epsilon}\right)\right)\dt s\right]\nonumber \\
 & = & \ep_{t,x}\left[\int_{t}^{t+\epsilon}h(s-t)\left(\hat{\pi}\left(s,X_{s}^{\hat{\pi}}\right)-\pi^{\epsilon}\left(s,X_{s}^{\epsilon}\right)\right)\dt s\right]\nonumber \\
 &  & +\ep_{t,x}\left[V^{\hat{\pi}}\left(t+\epsilon,X_{t+\epsilon}^{\hat{\pi}}\right)-V^{\hat{\pi}}\left(t+\epsilon,X_{t+\epsilon}^{\epsilon}\right)\right]\nonumber \\
 &  & +\ep_{t,x}\left[\int_{t+\epsilon}^{\tau_{t}^{\hat{\pi}}}\left(h\left(s-t\right)-h\left(s-t-\epsilon\right)\right)\hat{\pi}\left(s,X_{s}^{\hat{\pi}}\right)\dt s\right]\nonumber \\
 &  & -\ep_{t,x}\left[\int_{t+\epsilon}^{\tau_{t}^{\pi^{\epsilon}}}\left(h\left(s-t\right)-h\left(s-t-\epsilon\right)\right)\hat{\pi}\left(s,X_{s}^{\epsilon}\right)\dt s\right].\label{eq:4-5}
\end{eqnarray}
Here $\hat{\pi}(s,X_{s}^{\epsilon})$ and $\hat{\pi}(s,X_{s}^{\hat{\pi}})$
are the equilibrium control processes associated with the paths of
$X^{\epsilon}$ and $X^{\hat{\pi}}$, respectively.

According to the equation (\ref{eq:4-5}), we now consider the limitation
$\lim_{\epsilon\rightarrow0}\frac{V^{\hat{\pi}}(t,x)-V^{\pi^{\epsilon}}(t,x)}{\epsilon}$
in three parts separately:

1. Noting that $\int_{0}^{\infty}h(t)\dt t<\infty$, $l$ and $\hat{\pi}$
are bounded and applying the dominated convergence theorem, we get
\begin{align*}
\lim_{\epsilon\rightarrow0}\frac{\ep_{t,x}\left[\int_{t}^{t+\epsilon}h(s-t)\left(\hat{\pi}\left(s,X_{s}^{\hat{\pi}}\right)-\pi^{\epsilon}\left(s,X_{s}^{\epsilon}\right)\right)\dt s\right]}{\epsilon} & =\hat{\pi}\left(t,x\right)-\pi^{\epsilon}(t,x).
\end{align*}

2. We rewrite the second part in the right-side of the equation (\ref{eq:4-5})
by 
\begin{eqnarray*}
 &  & \ep_{t,x}\left[V^{\hat{\pi}}\left(t+\epsilon,X_{t+\epsilon}^{\hat{\pi}}\right)-V^{\hat{\pi}}\left(t+\epsilon,X_{t+\epsilon}^{\epsilon}\right)\right]\\
 & = & \ep_{t,x}\left[V^{\hat{\pi}}\left(t+\epsilon,X_{t+\epsilon}^{\hat{\pi}}\right)-V^{\hat{\pi}}\left(t,x\right)\right]-\ep_{t,x}\left[V^{\hat{\pi}}\left(t+\epsilon,X_{t+\epsilon}^{\epsilon}\right)-V^{\hat{\pi}}\left(t,x\right)\right]\\
 & = & \ep_{t,x}\left[\int_{t}^{t+\epsilon}\dt V^{\hat{\pi}}\left(u,X_{u}^{\hat{\pi}}\right)\right]-\ep_{t,x}\left[\int_{t}^{t+\epsilon}\dt V^{\hat{\pi}}\left(u,X_{u}^{\epsilon}\right)\right].
\end{eqnarray*}
Applying the Itô formula, we get 
\begin{eqnarray*}
 &  & \lim_{\epsilon\rightarrow0}\frac{\ep_{t,x}\left[\int_{t}^{t+\epsilon}\dt V^{\hat{\pi}}\left(u,X_{u}^{\hat{\pi}}\right)\right]}{\epsilon}\\
 & = & \frac{\partial V^{\hat{\pi}}(t,x)}{\partial t}+\left(\mu-\hat{\pi}\left(t,x\right)\right)\frac{\partial V^{\hat{\pi}}(t,x)}{\partial x}+\frac{1}{2}\sigma^{2}\frac{\partial^{2}V^{\hat{\pi}}(t,x)}{\partial x^{2}}\\
 & = & \left(\mcl L^{\hat{\pi}}V^{\hat{\pi}}\right)\left(t,x\right)\\
 & = & \left(\mcl L^{\hat{\pi}}V\right)\left(t,x\right),
\end{eqnarray*}
and
\begin{eqnarray*}
 &  & \lim_{\epsilon\rightarrow0}\frac{\ep_{t,x}\left[\int_{t}^{t+\epsilon}\dt V^{\hat{\pi}}\left(u,X_{u}^{\epsilon}\right)\right]}{\epsilon}\\
 & = & \frac{\partial V^{\hat{\pi}}(t,x)}{\partial t}+\left(\mu-l\right)\frac{\partial V^{\hat{\pi}}(t,x)}{\partial x}+\frac{1}{2}\sigma^{2}\frac{\partial^{2}V^{\hat{\pi}}(t,x)}{\partial x^{2}}\\
 & = & \left(\mcl L^{\pi^{\epsilon}}V^{\hat{\pi}}\right)\left(t,x\right)\\
 & = & \left(\mcl L^{\pi^{\epsilon}}V\right)\left(t,x\right).
\end{eqnarray*}

3. Considering the cases with $\tau_{t}^{\hat{\pi}}\geq\tau_{t}^{\pi^{\epsilon}}$
and $\tau_{t}^{\hat{\pi}}\leq\tau_{t}^{\pi^{\epsilon}}$and noting
that $\hat{\pi}\left(s,X_{s}^{\epsilon}\right)\equiv0$ for $s\geq\tau_{t}^{\pi^{\epsilon}}$,
we have
\begin{align*}
 & \ep_{t,x}\left[\int_{t+\epsilon}^{\tau_{t}^{\hat{\pi}}}\left(h\left(s-t\right)-h\left(s-t-\epsilon\right)\right)\hat{\pi}\left(s,X_{s}^{\hat{\pi}}\right)\dt s\right]\\
 & -\ep_{t,x}\left[\int_{t+\epsilon}^{\tau_{t}^{\pi^{\epsilon}}}\left(h\left(s-t\right)-h\left(s-t-\epsilon\right)\right)\hat{\pi}\left(s,X_{s}^{\epsilon}\right)\dt s\right]\\
\geq & \ep_{t,x}\left[\int_{t+\epsilon}^{\tau_{t}^{\hat{\pi}}}\left(h\left(s-t\right)-h\left(s-t-\epsilon\right)\right)\left[\hat{\pi}\left(s,X_{s}^{\hat{\pi}}\right)-\hat{\pi}\left(s,X_{s}^{\epsilon}\right)\right]\dt s\right].
\end{align*}
Noting that $\hat{\pi}$ is bounded and $\int_{0}^{\infty}h(s)\dt s<\infty,$
by the dominated convergence theorem, we get 
\[
\lim_{\epsilon\rightarrow0}\frac{\ep_{t,x}\left[\int_{t+\epsilon}^{\tau_{t}^{\hat{\pi}}}\left[h\left(s-t\right)-h\left(s-t-\epsilon\right)\right]\left(\hat{\pi}\left(s,X_{s}^{\hat{\pi}}\right)-\hat{\pi}\left(s,X_{s}^{\epsilon}\right)\right)\dt s\right]}{\epsilon}=0.
\]

Therefore, we obtain
\begin{align}
\lim_{\epsilon\rightarrow0}\frac{V^{\hat{\pi}}(t,x)-V^{\pi^{\epsilon}}(t,x)}{\epsilon} & \geq\left[\mcl L^{\hat{\pi}}V\left(t,x\right)+C\left(t,t,\hat{\pi}(t,x)\right)\right]-\left[\mcl L^{\pi^{\epsilon}}V\left(t,x\right)+C\left(t,t,\pi^{\epsilon}\left(t,x\right)\right)\right].\label{eq:4-3}
\end{align}
It follows from (\ref{eq:3-3}) and (\ref{eq:3-5}) that 
\begin{equation}
\left(\mcl L^{\hat{\pi}}V\right)(t,x)+C\left(t,t,\hat{\pi}(t,x)\right)=\sup_{\pi\in\Pi}\left\{ \left(\mcl L^{\pi}V\right)(t,x)+C\left(t,t,\pi(t,x)\right)\right\} .\label{eq:4-4}
\end{equation}
Therefore, (\ref{eq:4-3}) and (\ref{eq:4-4}) imply that 
\begin{eqnarray*}
\lim_{\epsilon\rightarrow0}\frac{V^{\hat{\pi}}(t,x)-V^{\pi^{\epsilon}}(t,x)}{\epsilon} & \geq & 0.
\end{eqnarray*}
This completes the proof.
\end{proof}

\section{\label{sec:5}Solutions to Two Special Cases }

In this section, we try to find a solution of the equilibrium HJB-equation
in Definition \ref{def:3-1} for specific discount functions. First
of all, we make a conjecture of equilibrium strategy for a general
discount function. Since 
\begin{eqnarray*}
H(s,t,l,p,P) & = & \frac{1}{2}\sigma^{2}P+(\mu-l)p+C(s,t,l)\\
 & = & \frac{1}{2}\sigma^{2}P+\mu p+[h(t-s)-p]l,
\end{eqnarray*}
we have
\[
\phi(s,t,p,P)=\begin{cases}
0, & \text{if }p\geq h(t-s),\\
M, & \text{if }p<h(t-s).
\end{cases}
\]

We assume that there exists a constant $b\geq0$ such that $\frac{\partial c}{\partial x}(t,t,x)\geq1,$
if $0\leq x<b,$ and $\frac{\partial c}{\partial x}(t,t,x)<1$, if
$x\geq b$. Thus, the equilibrium strategy is given by 
\begin{equation}
\hat{\pi}(t,x)=\phi\left(t,t,\frac{\partial c}{\partial x}(t,t,x),\frac{\partial^{2}c}{\partial x^{2}}(t,t,x)\right)=\begin{cases}
0,\q & \text{if }\;0\leq x<b,\\
M,\q & \text{if }\; x\geq b.
\end{cases}\label{eq:5-1}
\end{equation}
Then the equilibrium HJB-equation (\ref{eq:3-2}) becomes
\begin{equation}
\begin{cases}
\frac{\partial c}{\partial t}(s,t,x)+\frac{1}{2}\sigma^{2}\frac{\partial^{2}c}{\partial x^{2}}(s,t,x)+\mu\frac{\partial c}{\partial x}(s,t,x)=0, & (s,t,x)\in\mcl D[0,\infty)\times(0,b),\\
\frac{\partial c}{\partial t}(s,t,x)+\frac{1}{2}\sigma^{2}\frac{\partial^{2}c}{\partial x^{2}}(s,t,x)+(\mu-M)\frac{\partial c}{\partial x}(s,t,x)+h(t-s)M=0, & (s,t,x)\in\mcl D[0,\infty)\times[b,\infty),\\
c(s,t,0)=0,\q\forall(s,t)\in\mcl D[0,\infty).
\end{cases}\label{eq:5-2}
\end{equation}

\subsection{A Mixture of Exponential Discount Functions}

Let us consider a case where the dividends are proportionally paid
to $N$ inhomogenous shareholders. In terms of inhomogenous, we mean
that the shareholders have different discount rates. Then given a
control law $\pi$, the return function is 
\[
V^{\pi}(t,x)=\sum_{i=1}^{N}\ep_{t,x}\left[\int_{t}^{\tau_{t}^{\pi}}\omega_{i}e^{-\delta_{i}(z-t)}\pi(z,X_{z}^{\pi})\dt z\right],
\]
where $\omega_{i}>0$ satisfying $\sum_{i=1}^{N}\omega_{i}=1$ is
the proportion at which the dividends are paid to the shareholders,
$\delta_{i}>0$, $i=1,2,\cdots,N,$ are the constant discount rates
of the shareholders, respectively.

In fact, a mixture of exponential discount functions is used in the
above example. We consider a discount function defined by
\begin{equation}
h(t)=\sum_{i=1}^{N}\omega_{i}e^{-\delta_{i}t},\q t\geq0,\label{eq:5-1-1}
\end{equation}
where $\delta_{i}>0$, and $\omega_{i}>0$ satisfies $\sum_{i=1}^{N}\omega_{i}=1.$ 

We consider the following ansatz:
\begin{eqnarray}
c(s,t,x) & = & \sum_{i=1}^{N}\omega_{i}e^{-\delta_{i}(t-s)}V_{i}(x),\q(s,t,x)\in\mcl D[0,\infty)\times[0,\infty),\label{eq:5-1-2}
\end{eqnarray}
where the functions $V_{i}(x)$, $i=1,2,\cdots,N,$ are given by the
system of ODEs
\begin{equation}
\begin{cases}
\frac{1}{2}\sigma^{2}\frac{\partial^{2}V_{i}}{\partial x^{2}}(x)+\mu\frac{\partial V_{i}}{\partial x}(x)-\delta_{i}V_{i}(x)=0, & x\in[0,b),\\
\frac{1}{2}\sigma^{2}\frac{\partial^{2}V_{i}}{\partial x^{2}}(x)+\left(\mu-M\right)\frac{\partial V_{i}}{\partial x}(x)-\delta_{i}V_{i}(x)+M=0, & x\in[b,\infty),\\
V_{i}(0)=0.
\end{cases}\label{eq:5-1-3}
\end{equation}

Denote by $\theta_{1}(\eta,c)$ and $-\theta_{2}(\eta,c)$ the positive
and negative roots of the equation $\frac{1}{2}\sigma^{2}y^{2}+\eta y-c=0$,
respectively. Then
\[
\begin{cases}
\theta_{1}(\eta,c) & =\frac{-\eta+\sqrt{\eta^{2}+2\sigma^{2}c}}{\sigma^{2}},\\
\theta_{2}(\eta,c) & =\frac{\eta+\sqrt{\eta^{2}+2\sigma^{2}c}}{\sigma^{2}}.
\end{cases}
\]
Thus a general solution of the equation (\ref{eq:5-1-3}) has the
form

\begin{equation}
V_{i}(x)=\begin{cases}
C_{i1}e^{\theta_{1}(\mu,\delta_{i})x}+C_{i2}e^{-\theta_{2}(\mu,\delta_{i})x}, & x\in[0,b),\\
\frac{M}{\delta_{i}}+C_{i3}e^{\theta_{1}(\mu-M,\delta_{i})x}+C_{i4}e^{-\theta_{2}(\mu-M,\delta_{i})x}, & x\in[b,\infty),
\end{cases}\label{eq:5-1-4}
\end{equation}
for $i=1,2,\cdots,N$.

Since $V_{i}(0)=0,$ and $V_{i}(x)>0,$ for all $x>0$, we have $C_{i1}=-C_{i2}:=C_{i}>0$,
$i=1,2,\cdots,N$. Since we are looking for a bounded function $c(\cdot,\cdot,\cdot)$
(see Theorem \ref{thm:3-2}), we have $C_{i3}=0$, $i=1,2,\cdots,N.$
To simplify the notation, let $C_{i4}:=-d_{i},i=1,2,\cdots,N.$

Now to find the value of $C_{i},d_{i},i=1,2,\cdots,N$ and $b$, we
use ``the principle of smooth fit'' to get
\begin{equation}
\begin{cases}
V_{i}(b+) & =V_{i}(b-),\q i=1,2,\cdots,N,\\
V_{i}'(b+) & =V_{i}'(b-),\q i=1,2,\cdots,N,\\
\frac{\partial c}{\partial x}(t,t,b+) & =1\ \left(\mathrm{or}\ \mathrm{equivalently},\ \frac{\partial c}{\partial x}(t,t,b-)=1\right).
\end{cases}\label{eq:5-1-5}
\end{equation}
Therefore by denoting 
\begin{align*}
\theta_{i1}=\theta_{1}(\mu,\delta_{i}),\ \theta_{i2}=\theta_{2}(\mu,\delta_{i}),\  & \theta_{i3}=\theta_{2}(\mu-M,\delta_{i}),\q i=1,2,\cdots,N,
\end{align*}
we can rewrite (\ref{eq:5-1-5}) as for $i=1,2,\cdots,N,$ 
\begin{eqnarray}
C_{i}\left(e^{\theta_{i1}b}-e^{-\theta_{i2}b}\right) & = & \frac{M}{\delta_{i}}-d_{i}e^{-\theta_{i3}b},\label{eq:5-1-6}\\
C_{i}\left(\theta_{i1}e^{\theta_{i1}b}+\theta_{i2}e^{-\theta_{i2}b}\right) & = & d_{i}\theta_{i3}e^{-\theta_{i3}b},\label{eq:5-1-7}
\end{eqnarray}
and
\begin{equation}
\sum_{i=1}^{N}\omega_{i}C_{i}\left(\theta_{i1}e^{\theta_{i1}b}+\theta_{i2}e^{-\theta_{i2}b}\right)=1.\label{eq:5-1-8}
\end{equation}
From (\ref{eq:5-1-6}) - (\ref{eq:5-1-7}) we can get $C_{i}$ and
$d_{i}$ in the expression of $b$:
\begin{align}
C_{i} & =\frac{M\theta_{i3}}{\delta_{i}}\left[\left(\theta_{i1}+\theta_{i3}\right)e^{\theta_{i1}b}+\left(\theta_{i2}-\theta_{i3}\right)e^{-\theta_{i2}b}\right]^{-1},\label{eq:5-1-9}\\
d_{i} & =\frac{M}{\delta_{i}}e^{\theta_{i3}b}\frac{\theta_{i1}e^{\theta_{i1}b}+\theta_{i2}e^{-\theta_{i2}b}}{\left(\theta_{i1}+\theta_{i3}\right)e^{\theta_{i1}b}+\left(\theta_{i2}-\theta_{i3}\right)e^{-\theta_{i2}b}},\label{eq:5-1-10}
\end{align}
for $i=1,2,\cdots,N$.

Substituting $C_{i}$ into (\ref{eq:5-1-8}), we obtain
\[
\sum_{i=1}^{N}\omega_{i}\frac{M\theta_{i3}}{\delta_{i}}\frac{\theta_{i1}e^{\theta_{i1}b}+\theta_{i2}e^{-\theta_{i2}b}}{\left(\theta_{i1}+\theta_{i3}\right)e^{\theta_{i1}b}+\left(\theta_{i2}-\theta_{i3}\right)e^{-\theta_{i2}b}}=1.
\]
Let 
\[
F(b):=\sum_{i=1}^{N}\omega_{i}\frac{M\theta_{i3}}{\delta_{i}}\frac{\theta_{i1}e^{\theta_{i1}b}+\theta_{i2}e^{-\theta_{i2}b}}{\left(\theta_{i1}+\theta_{i3}\right)e^{\theta_{i1}b}+\left(\theta_{i2}-\theta_{i3}\right)e^{-\theta_{i2}b}}-1.
\]

\begin{lem}
If $\sum_{i=1}^{N}\omega_{i}\frac{M\theta_{i3}}{\delta_{i}}>1$, then
$F(b)=0$ has a unique positive solution.\end{lem}
\begin{proof}
The condition $\sum_{i=1}^{N}\omega_{i}\frac{M\theta_{i3}}{\delta_{i}}>1$
implies that $F(0)>0.$ From Lemma 2.1 of \citet{at97}, we know that
\[
\frac{M}{\delta_{i}}-\frac{1}{\theta_{i3}}-\frac{1}{\theta_{i1}}<0,\q i=1,2,\cdots,N.
\]
Thus, 
\begin{eqnarray*}
F(+\infty) & = & \sum_{i=1}^{N}\omega_{i}\frac{M\theta_{i3}}{\delta_{i}}\frac{\theta_{i1}}{\theta_{i1}+\theta_{i3}}-1\\
 & = & \sum_{i=1}^{N}\omega_{i}\left(\frac{M\theta_{i3}}{\delta_{i}}\frac{\theta_{i1}}{\theta_{i1}+\theta_{i3}}-1\right)\\
 & = & \sum_{i=1}^{N}\omega_{i}\frac{\theta_{i1}\theta_{i3}}{\theta_{i1}+\theta_{i3}}\left(\frac{M}{\delta_{i}}-\frac{1}{\theta_{i3}}-\frac{1}{\theta_{i1}}\right)\\
 & < & 0.
\end{eqnarray*}
Furthermore, we have
\[
F'(b)=\sum_{i=1}^{N}\omega_{i}\frac{M\theta_{i3}}{\delta_{i}}\frac{\Delta_{i}}{\left[\left(\theta_{i1}+\theta_{i3}\right)e^{\theta_{i1}b}+\left(\theta_{i2}-\theta_{i3}\right)e^{-\theta_{i2}b}\right]^{2}},
\]
where
\begin{eqnarray*}
\Delta_{i} & = & \left(\theta_{i1}^{2}e^{\theta_{i1}b}-\theta_{i2}^{2}e^{-\theta_{i2}b}\right)\left[\left(\theta_{i1}+\theta_{i3}\right)e^{\theta_{i1}b}+\left(\theta_{i2}-\theta_{i3}\right)e^{-\theta_{i2}b}\right]\\
 &  & -\left(\theta_{i1}e^{\theta_{i1}b}+\theta_{i2}e^{-\theta_{i2}b}\right)\left[\theta_{i1}\left(\theta_{i1}+\theta_{i3}\right)e^{\theta_{i1}b}-\theta_{i2}\left(\theta_{i2}-\theta_{i3}\right)e^{-\theta_{i2}b}\right]\\
 & = & \left[\theta_{i1}^{2}\left(\theta_{i2}-\theta_{i3}\right)-\theta_{i2}^{2}\left(\theta_{i1}+\theta_{i3}\right)+\theta_{i1}\theta_{i2}\left(\theta_{i2}-\theta_{i3}\right)-\theta_{i2}\theta_{i1}\left(\theta_{i1}+\theta_{i3}\right)\right]e^{\left(\theta_{i1}-\theta_{i2}\right)b}\\
 & = & \left[\theta_{i1}\left(\theta_{i2}-\theta_{i3}\right)-\theta_{i2}\left(\theta_{i1}+\theta_{i3}\right)\right]\left(\theta_{i1}+\theta_{i2}\right)e^{\left(\theta_{i1}-\theta_{i2}\right)b}\\
 & = & -\theta_{i3}\left(\theta_{i1}+\theta_{i2}\right)^{2}e^{\left(\theta_{i1}-\theta_{i2}\right)b}\\
 & < & 0.
\end{eqnarray*}
Therefore, the equation $F(b)=0$ admits a unique solution on $(0,\infty).$\end{proof}
\begin{thm}
\label{thm:4-2}Given the discount function (\ref{eq:5-1-1}), there
exists a smooth function $c(\cdot,\cdot,\cdot)$ satisfying the equilibrium
HJB-equation (\ref{eq:3-2}).
\begin{enumerate}[label=(\roman{enumi})]
\item If $\sum_{i=1}^{N}\omega_{i}\frac{M\theta_{i3}}{\delta_{i}}\leq1$,
then \textbf{$b=0$} and the function $c(\cdot,\cdot,\cdot)$ is given
by
\begin{equation}
c(s,t,x)=\sum_{i=1}^{N}\omega_{i}e^{-\delta_{i}(t-s)}\frac{M}{\delta_{i}}\left(1-e^{-\theta_{i3}x}\right),\q x\in[0,\infty).\label{eq:5-1-11}
\end{equation}

\item If $\sum_{i=1}^{N}\omega_{i}\frac{M\theta_{i3}}{\delta_{i}}>1$, then
\begin{equation}
c(s,t,x)=\begin{cases}
\sum_{i=1}^{N}\omega_{i}e^{-\delta_{i}(t-s)}C_{i}\left(e^{\theta_{i1}x}-e^{-\theta_{i2}x}\right), & \q x\in[0,b),\\
\sum_{i=1}^{N}\omega_{i}e^{-\delta_{i}(t-s)}\left(\frac{M}{\delta_{i}}-d_{i}e^{-\theta_{i3}x}\right), & \q x\in[b,\infty),
\end{cases}\label{eq:5-1-12}
\end{equation}
where $C_{i},d_{i},i=1,2,\cdots,N,$ and $b$ is the unique solution
to the system (\ref{eq:5-1-6})-(\ref{eq:5-1-8}).
\end{enumerate}
\end{thm}
\begin{proof}
(i) It is easy to check the function $c(s,t,x)$ given by (\ref{eq:5-1-11})
and $b=0$ satisfy the system of ODEs (\ref{eq:5-2}). Obviously,
we have 
\begin{eqnarray*}
\frac{\partial c}{\partial x}(s,t,0) & = & \sum_{i=1}^{N}\omega_{i}e^{-\delta_{i}(t-s)}\frac{M}{\delta_{i}}\theta_{i3}\leq1,\q(s,t)\in\mcl D[0,\infty),\\
\frac{\partial^{2}c}{\partial x^{2}}(s,t,x) & = & -\sum_{i=1}^{N}\omega_{i}e^{-\delta_{i}(t-s)}\frac{M}{\delta_{i}}\theta_{i3}^{2}e^{-\theta_{i3}x}<0,\q(s,t,x)\in\mcl D[0,\infty)\times[0,\infty).
\end{eqnarray*}
Thus, $\frac{\partial c}{\partial x}(t,t,x)<1,$ for $x\geq0$, which
implies $c(\cdot,\cdot,\cdot)$ satisfies the equilibrium HJB-equation
(\ref{eq:3-2}).

(ii) Similarly, it is easy to check that $b$ and $c(\cdot,\cdot,\cdot,)$
given by (\ref{eq:5-1-6})-(\ref{eq:5-1-8}) and (\ref{eq:5-1-12})
satisfy the system of ODEs (\ref{eq:5-2}). It is sufficient to show
\begin{equation}
\begin{cases}
\frac{\partial c}{\partial x}(t,t,x)\geq1, & x\in[0,b),\\
\frac{\partial c}{\partial x}(t,t,x)<1, & x\in[b,\infty).
\end{cases}\label{eq:55-1-13}
\end{equation}
The first and second derivatives of $c(s,t,x)$ given by (\ref{eq:5-1-12})
with respective to $x$ are 
\[
\frac{\partial c}{\partial x}(t,t,x)=\begin{cases}
\sum_{i=1}^{N}\omega_{i}C_{i}\left(\theta_{i1}e^{\theta_{i1}x}+\theta_{i2}e^{-\theta_{i2}x}\right), & (t,x)\in[0,\infty)\times[0,b),\\
\sum_{i=1}^{N}\omega_{i}d_{i}\theta_{i3}e^{-\theta_{i3}x}, & (t,x)\in[0,\infty)\times[b,\infty),
\end{cases}
\]
 and

\[
\frac{\partial^{2}c}{\partial x^{2}}(t,t,x)=\begin{cases}
\sum_{i=1}^{N}\omega_{i}C_{i}\left(\theta_{i1}^{2}e^{\theta_{i1}x}-\theta_{i2}^{2}e^{-\theta_{i2}x}\right), & (t,x)\in[0,\infty)\times[0,b),\\
-\sum_{i=1}^{N}\omega_{i}d_{i}\theta_{i3}^{2}e^{-\theta_{i3}x}, & (t,x)\in[0,\infty)\times[b,\infty),
\end{cases}
\]
respectively. 

It is easy to check that $\frac{\partial c}{\partial x}(t,t,x)>0$,
for all $(t,x)\in[0,\infty)\times[0,\infty)$, which implies that
$c(t,t,\cdot)$ is strictly increasing. Next we show that $c(t,t,\cdot)$
is a concave function on $[0,\infty)$, i.e. $\frac{\partial^{2}c}{\partial x^{2}}(t,t,x)<0$,
for all $(t,x)\in[0,\infty)\times[0,\infty)$. First we show that
$\frac{\partial^{2}c}{\partial x^{2}}(t,t,x)$ is continuous at $x=b$.
Apparently, $\frac{\partial^{2}c}{\partial x^{2}}(t,t,x)<0$, for
all $(t,x)\in[0,\infty)\times[b,\infty)$. Recalling (\ref{eq:5-1-2}),
(\ref{eq:5-1-3}) and (\ref{eq:5-1-5}), we have
\begin{align*}
\frac{1}{2}\sigma^{2}\frac{\partial^{2}c}{\partial x^{2}}(t,t,b-) & =-\mu\frac{\partial c}{\partial x}(t,t,b)+\sum_{i=1}^{N}\omega_{i}\delta_{i}V_{i}(b),\\
\frac{1}{2}\sigma^{2}\frac{\partial^{2}c}{\partial x^{2}}(t,t,b+) & =-\left(\mu-M\right)\frac{\partial c}{\partial x}(t,t,b)+\sum_{i=1}^{N}\omega_{i}\delta_{i}V_{i}(b)-M.
\end{align*}
Since $\frac{\partial c}{\partial x}(t,t,b)=1$, we get $\frac{\partial^{2}c}{\partial x^{2}}(t,t,b-)=\frac{\partial^{2}c}{\partial x^{2}}(t,t,b+)=\frac{\partial^{2}c}{\partial x^{2}}(t,t,b).$

Obviously, for all $0\leq x\leq b,$ we have 
\[
\frac{\partial^{3}c}{\partial x^{3}}(t,t,x)=\sum_{i=1}^{N}\omega_{i}C_{i}\left(\theta_{i1}^{3}e^{\theta_{i1}x}+\theta_{i2}^{3}e^{-\theta_{i2}x}\right)>0,
\]
which means that $\frac{\partial^{2}c}{\partial x^{2}}(t,t,x)\leq\frac{\partial^{2}c}{\partial x^{2}}(t,t,b)<0,$
for all $0\leq x\leq b$. Thus, we proved (\ref{eq:55-1-13}).\end{proof}
\begin{cor}
\label{cor:4-3}Consider the discount function (\ref{eq:5-1-1}).
\begin{enumerate}[label=(\roman{enumi})]
\item If $\sum_{i=1}^{N}\omega_{i}\frac{M\theta_{i3}}{\delta_{i}}\leq1$,
then for $t\in[0,\infty)$
\[
\hat{\pi}(t,x)=\phi\left(t,t,\frac{\partial c}{\partial x}(t,t,x),\frac{\partial^{2}c}{\partial x^{2}}(t,t,x)\right)=M,\q x\in[0,\infty),
\]
is an equilibrium dividend strategy, and 
\[
V(t,x)=c(t,t,x)=\sum_{i=1}^{N}\omega_{i}\frac{M}{\delta_{i}}\left(1-e^{-\theta_{i3}x}\right),\q x\in[0,\infty),
\]
is the corresponding equilibrium value function.
\item If $\sum_{i=1}^{N}\omega_{i}\frac{M\theta_{i3}}{\delta_{i}}>1$, then
for $t\in[0,\infty)$
\[
\hat{\pi}(t,x)=\phi\left(t,t,\frac{\partial c}{\partial x}(t,t,x),\frac{\partial^{2}c}{\partial x^{2}}(t,t,x)\right)=\begin{cases}
0, & x\in[0,b),\\
M, & x\in[b,\infty),
\end{cases}
\]
is an equilibrium dividend strategy, and 
\[
V(t,x)=c(t,t,x)=\begin{cases}
\sum_{i=1}^{N}\omega_{i}C_{i}\left(e^{\theta_{i1}x}-e^{-\theta_{i2}x}\right), & x\in[0,b),\\
\sum_{i=1}^{N}\omega_{i}\left(\frac{M}{\delta_{i}}-d_{i}e^{-\theta_{i3}x}\right), & x\in[b,\infty),
\end{cases}
\]
is the corresponding equilibrium value function. Here $C_{i},d_{i},i=1,2,\cdots,N,$
and $b$ is the unique solution to the system (\ref{eq:5-1-6})-(\ref{eq:5-1-8}).
\end{enumerate}
\end{cor}
\begin{proof}
By Theorem \ref{thm:3-2} and Theorem \ref{thm:4-2}, it is sufficient
to verify (\ref{eq:3-4}). If $M\geq\mu$, in both cases (i) and (ii),
it is well known that $\pr\left(\tau_{t}^{\hat{\pi}}<\infty\right)=1$
(see, e.g. \citet{gs06b}). Since $c(s,t,0)=0$ for all $(s,t)\in\mcl D[0,\infty),$we
get (\ref{eq:3-4}). If $M<\mu$, in both cases (i) and (ii), we have
$\pr\left(\tau_{t}^{\hat{\pi}}=\infty\right)>0$ and $X_{\tau_{t}^{\hat{\pi}}}^{\hat{\pi}}=+\infty$$ $
on $\{\tau_{t}^{\hat{\pi}}=\infty\}$. However, for any $s\in[0,\infty)$
we have $\lim_{t\tto\infty,x\tto\infty}c(s,t,x)=0$. Thus, we still
have (\ref{eq:3-4}).\end{proof}
\begin{example}
Let $N=2,$ $\mu$ = 1, $\sigma$ = 1, $M$ = 0.8, $\delta_{1}$ =
0.2, $\delta_{2}$ = 0.4. Figure \ref{fig:4-1} illustrates the equilibrium
value functions for the mixture of exponential discount functions
with $\omega=0$, 0.4, 0.7 and 1. The barriers are $0.6525$, $0.8781$,
$1.0207$ and $1.1452$, respectively. The cases with $\omega=0$
and 1 are time consistent and the equilibrium strategies are optimal. 
\end{example}
\begin{figure}
\begin{centering}
\includegraphics[scale=1.2]{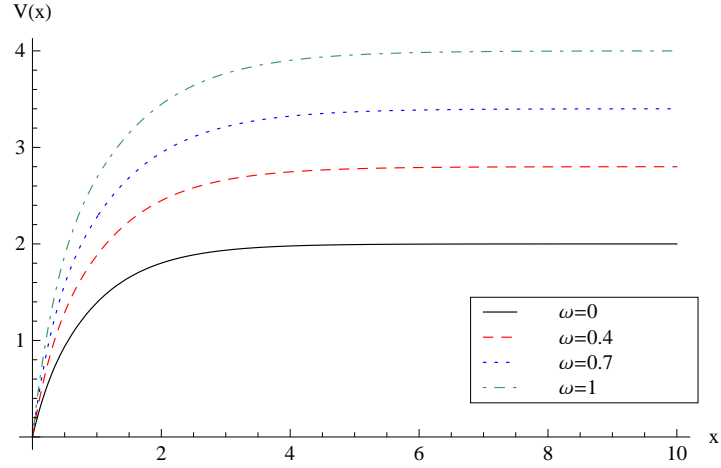}
\par\end{centering}

\caption{\label{fig:4-1}Equilibrium value functions with a mixture of exponential
discount functions }
\end{figure}

\subsection{A Pseudo-Exponential Discount Function}

We now consider a pseudo-exponential discount function defined as
\begin{equation}
h(t)=(1+\lambda t)e^{-\delta t},\q t\geq0,\label{eq:4-2-1}
\end{equation}
where $\lambda>0,\ \delta>0$ are parameters. We refer the reader
to \citet{ep08} for explanations of this discount function. To ensure
$h$ is decreasing, we assume that $\lambda<\delta$. To simplify
the calculations, we shall impose more conditions on $\lambda$ in
the following.

We consider the following ansatz:
\begin{equation}
c(s,t,x)=e^{-\delta(t-s)}\left\{ \lambda(t-s)V_{3}(x)+V_{4}(x)\right\} ,\q(s,t,x)\in\mcl D[0,\infty)\times[0,\infty),\label{eq:4-2-2}
\end{equation}
where $V_{3}(\cdot)$ and $V_{4}(\cdot)$ are given by 
\begin{equation}
\begin{cases}
\frac{1}{2}\sigma^{2}\frac{\partial^{2}V_{3}}{\partial x^{2}}(x)+\mu\frac{\partial V_{3}}{\partial x}(x)-\delta V_{3}(x)=0, & x\in[0,b),\\
\frac{1}{2}\sigma^{2}\frac{\partial^{2}V_{3}}{\partial x^{2}}(x)+\left(\mu-M\right)\frac{\partial V_{3}}{\partial x}(x)-\delta V_{3}(x)+M=0, & x\in[b,\infty),\\
V_{3}(0)=0,
\end{cases}\label{eq:4-2-3}
\end{equation}
and 
\begin{equation}
\begin{cases}
\frac{1}{2}\sigma^{2}\frac{\partial^{2}V_{4}}{\partial x^{2}}(x)+\mu\frac{\partial V_{4}}{\partial x}(x)-\delta V_{4}(x)+\lambda V_{3}(x)=0, & x\in[0,b),\\
\frac{1}{2}\sigma^{2}\frac{\partial^{2}V_{4}}{\partial x^{2}}(x)+(\mu-M)\frac{\partial V_{4}}{\partial x}(x)-\delta V_{4}(x)+\lambda V_{3}(x)+M=0, & x\in[b,\infty),\\
V_{4}(0)=0,
\end{cases}\label{eq:4-2-4}
\end{equation}
respectively. It is easy to check that the function $c(\cdot,\cdot,\cdot)$
given by (\ref{eq:4-2-2})-(\ref{eq:4-2-4}) satisfies the system
(\ref{eq:5-2}).

Recalling the situation we discussed in Subsection 4.1, the equation
(\ref{eq:4-2-3}) has a general solution
\begin{equation}
V_{3}(x)=\begin{cases}
C\left(e^{\theta_{1}(\mu)x}-e^{-\theta_{2}(\mu)x}\right), & x\in[0,b),\\
\frac{M}{\delta}-de^{-\theta_{2}(\mu-M)x}, & x\in[b,\infty),
\end{cases}\label{eq:4-2-5}
\end{equation}
where $C>0,\ d>0$ are two unknown constants to be determined, $\theta_{1}(\eta)$
and $-\theta_{2}(\eta)$ are the positive and negative roots of the
equation $\frac{1}{2}\sigma^{2}y^{2}+\eta y-\delta=0$, respectively. 

According to ``the principle of smooth fit'', we have
\begin{equation}
\begin{cases}
V_{3}(b+) & =V_{3}(b-),\\
V_{3}'(b+) & =V_{3}'(b-),
\end{cases}\label{eq:4-2-6}
\end{equation}
which yields that
\begin{align}
C & =\frac{M\theta_{3}}{\delta}\left[\left(\theta_{1}+\theta_{3}\right)e^{\theta_{1}b}+\left(\theta_{2}-\theta_{3}\right)e^{-\theta_{2}b}\right]^{-1},\label{eq:4-2-0}\\
d & =\frac{M}{\delta}e^{\theta_{3}b}\frac{\theta_{1}e^{\theta_{1}b}+\theta_{2}e^{-\theta_{2}b}}{\left(\theta_{1}+\theta_{3}\right)e^{\theta_{1}b}+\left(\theta_{2}-\theta_{3}\right)e^{-\theta_{2}b}},\label{eq:4-2-00}
\end{align}
where 
\[
\theta_{1}=\theta_{1}(\mu),\q\theta_{2}=\theta_{2}(\mu),\q\theta_{3}=\theta_{2}(\mu-M).
\]

After obtaining $V_{3}$, solving ODE (\ref{eq:4-2-4}) yields that
\begin{equation}
V_{4}(x)=\begin{cases}
\left(D_{1}-B_{1}x\right)e^{\theta_{1}x}+\left(D_{2}+B_{2}x\right)e^{-\theta_{2}x}, & \q0\leq x<b,\\
\frac{M}{\delta}\left(1+\frac{\lambda}{\delta}\right)+\left(D_{3}+B_{3}x\right)e^{-\theta_{3}x}, & \q x\geq b,
\end{cases}\label{eq:4-2-7}
\end{equation}
where
\begin{equation}
B_{1}=\frac{\lambda C}{\mu+\sigma^{2}\theta_{1}}>0,\q B_{2}=\frac{\lambda C}{\mu-\sigma^{2}\theta_{2}}<0,\q B_{3}=\frac{\lambda d}{\mu-M-\sigma^{2}\theta_{3}}<0.\label{eq:4-2-000}
\end{equation}
 Since $V_{4}(0)=0$, we have $D_{1}=-D_{2}:=\hat{C}$. Also noting
that $B_{1}+B_{2}=0$, we rewrite (\ref{eq:4-2-7}) as 
\begin{equation}
V_{4}(x)=\begin{cases}
\left(\hat{C}-B_{1}x\right)e^{\theta_{1}x}-\left(\hat{C}+B_{1}x\right)e^{-\theta_{2}x}, & \q0\leq x<b,\\
\frac{M}{\delta}\left(1+\frac{\lambda}{\delta}\right)+\left(D_{3}+B_{3}x\right)e^{-\theta_{3}x}, & \q x\geq b.
\end{cases}\label{eq:4-2-8}
\end{equation}

Applying \textquotedbl{}the principle of smooth fit\textquotedbl{},
we obtain
\begin{equation}
\begin{cases}
V_{4}(b+) & =V_{4}(b-),\\
V_{4}'(b+) & =V_{4}'(b-),\\
\frac{\partial c}{\partial x}(t,t,b+) & =1\ \left(\text{or equavalently, }\frac{\partial c}{\partial x}(t,t,b-)=1\right).
\end{cases}\label{eq:4-2-9}
\end{equation}
From the first two equations in (\ref{eq:4-2-9}), we obtain
\begin{align}
\hat{C} & =\frac{\left[\left(\theta_{1}+\theta_{3}\right)b+1\right]B_{1}e^{\theta_{1}b}-\left[\left(\theta_{2}-\theta_{3}\right)b-1\right]B_{1}e^{-\theta_{2}b}+B_{3}e^{-\theta_{3}b}+\theta_{3}\left(1+\frac{\lambda}{\delta}\right)\frac{M}{\delta}}{\left(\theta_{1}+\theta_{3}\right)e^{\theta_{1}b}+\left(\theta_{2}-\theta_{3}\right)e^{-\theta_{2}b}},\label{eq:4-2-10}\\
D_{3} & =e^{\theta_{3}b}\left[\left(\hat{C}-B_{1}b\right)e^{\theta_{1}b}-\left(\hat{C}+B_{1}b\right)e^{-\theta_{2}b}-\left(1+\frac{\lambda}{\delta}\right)\frac{M}{\delta}\right]-B_{3}b.\label{eq:4-2-11}
\end{align}
Furthermore, using $\frac{\partial c}{\partial x}(t,t,b+)=\frac{\partial c}{\partial x}(t,t,b-)=1$,
we have
\begin{equation}
\left[\left(\hat{C}-B_{1}b\right)\theta_{1}-B_{1}\right]e^{\theta_{1}b}+\left[\left(\hat{C}+B_{1}b\right)\theta_{2}-B_{1}\right]e^{-\theta_{2}b}-1=0,\label{eq:4-2-12}
\end{equation}
i.e.,
\[
\hat{C}=\frac{1+\left(\theta_{1}b+1\right)B_{1}e^{\theta_{1}b}-\left(\theta_{2}b-1\right)B_{1}e^{-\theta_{2}b}}{\theta_{1}e^{\theta_{1}b}+\theta_{2}e^{-\theta_{2}b}},
\]
and
\[
\left(-\theta_{3}D_{3}-\theta_{3}B_{3}b+B_{3}\right)e^{-\theta_{3}b}-1=0,
\]
i.e., 
\begin{equation}
D_{3}=\frac{1}{\theta_{3}}\left(B_{3}-e^{\theta_{3}b}\right)-B_{3}b.\label{eq:4-2-13}
\end{equation}

Puting (\ref{eq:4-2-10}) and (\ref{eq:4-2-11}) into the left-hand-side
of (\ref{eq:4-2-12}), it can be rewritten as
\[
\left[\left(\theta_{1}+\theta_{3}\right)e^{\theta_{1}b}+\left(\theta_{2}-\theta_{3}\right)e^{-\theta_{2}b}\right]^{-1}G(b),
\]
where
\begin{eqnarray}
G(b) & := & -\theta_{3}B_{1}e^{2\theta_{1}b}+\theta_{3}B_{1}e^{-2\theta_{2}b}+\theta_{1}B_{3}e^{\left(\theta_{1}-\theta_{3}\right)b}\nonumber \\
 &  & +\theta_{2}B_{3}e^{-\left(\theta_{2}+\theta_{3}\right)b}+2\left(\theta_{1}+\theta_{2}\right)\theta_{3}B_{1}be^{\left(\theta_{1}-\theta_{2}\right)b}\nonumber \\
 &  & +\left[\theta_{1}\theta_{3}\left(1+\frac{\lambda}{\delta}\right)\frac{M}{\delta}-\left(\theta_{1}+\theta_{3}\right)\right]e^{\theta_{1}b}\nonumber \\
 &  & +\left[\theta_{2}\theta_{3}\left(1+\frac{\lambda}{\delta}\right)\frac{M}{\delta}-\left(\theta_{2}-\theta_{3}\right)\right]e^{-\theta_{2}b},\label{eq:4-2-14}
\end{eqnarray}
and
\[
G(0)=\left(\theta_{1}+\theta_{2}\right)\left\{ \left[\frac{\lambda}{\mu-M-\sigma^{2}\theta_{3}}+\theta_{3}\left(1+\frac{\lambda}{\delta}\right)\right]\frac{M}{\delta}-1\right\} .
\]

\begin{lem}
\label{lem:4-5}If 
\begin{equation}
\frac{\frac{\delta}{M}-\theta_{3}}{\frac{1}{\mu-M-\sigma^{2}\theta_{3}}+\frac{\theta_{3}}{\delta}}<\lambda<\frac{\left(\theta_{1}+\theta_{3}\right)\left[\frac{\delta}{M}\left(\theta_{1}+\theta_{3}\right)-\theta_{1}\theta_{3}\right]}{\theta_{1}^{2}\left(\frac{1}{\mu-M-\sigma^{2}\theta_{3}}+\frac{\theta_{3}}{\delta}\right)+\theta_{3}^{2}\left(\frac{\theta_{1}}{\delta}-\frac{1}{\mu+\sigma^{2}\theta_{1}}\right)},\label{eq:4-2-16}
\end{equation}
then $G(b)=0$ has a positive solution.
\end{lem}
The proof of Lemma \ref{lem:4-5} is shown in Appendix A. Now we show
the main result of this subsection in the following theorem.
\begin{thm}
\label{thm:4-6}Assume that $0<\lambda<\delta$. Given the discount
function (\ref{eq:4-2-1}), there exists a smooth function $c(\cdot,\cdot,\cdot)$
satisfying the equilibrium HJB-equation (\ref{eq:3-2}). 
\begin{enumerate}[label=(\roman{enumi})]
\item If $\frac{\delta}{M}>\theta_{3}$ and $\lambda\leq\left(\frac{\delta}{M}-\theta_{3}\right)\left[\frac{1}{\mu-M-\sigma^{2}\theta_{3}}+\frac{\theta_{3}}{\delta}\right]^{-1}$,
then $b=0$ and $c(\cdot,\cdot,\cdot)$ is given by (\ref{eq:4-2-2})
with 
\begin{align}
V_{3}(x) & =\frac{M}{\delta}\left(1-e^{-\theta_{3}x}\right),\q x\in[0,\infty),\nonumber \\
V_{4}(x) & =\left(1+\frac{\lambda}{\delta}\right)\frac{M}{\delta}+\frac{M}{\delta}\left[\frac{\lambda}{\mu-M-\sigma^{2}\theta_{3}}x-\left(1+\frac{\lambda}{\delta}\right)\right]e^{-\theta_{3}x},\q x\in[0,\infty).\label{eq:4-2-17}
\end{align}

\item If (\ref{eq:4-2-16}) and (\ref{eq:b-1}) hold, then $c(\cdot,\cdot,\cdot)$
is given by (\ref{eq:4-2-2}) with 
\begin{align}
V_{3}(x) & =\begin{cases}
C\left(e^{\theta_{1}x}-e^{-\theta_{2}x}\right), & x\in[0,b),\\
\frac{M}{\delta}-de^{-\theta_{3}x}, & x\in[b,\infty),
\end{cases}\nonumber \\
V_{4}(x) & =\begin{cases}
\left(\hat{C}-B_{1}x\right)e^{\theta_{1}x}-\left(\hat{C}+B_{1}x\right)e^{-\theta_{2}x}, & 0\leq x<b,\\
\left(1+\frac{\lambda}{\delta}\right)\frac{M}{\delta}+\left(D_{3}+B_{3}x\right)e^{-\theta_{3}x}, & x\in[b,\infty),
\end{cases}\label{eq:4-2-18}
\end{align}
where $\left(b,C,d,\hat{C},B_{1},B_{3},D_{3}\right)$ is a solution
to (\ref{eq:4-2-6}) and (\ref{eq:4-2-9}).
\end{enumerate}
\end{thm}
\begin{proof}
It is easy to check that the function $c(\cdot,\cdot,\cdot)$ given
by (\ref{eq:4-2-2})-(\ref{eq:4-2-4}) satisfies the system (\ref{eq:5-2}).
To prove $c(\cdot,\cdot,\cdot)$ satisfies the equilibrium HJB-equation
(\ref{eq:3-2}), it is sufficient to show 
\begin{equation}
\begin{cases}
\frac{\partial c}{\partial x}(t,t,x)\geq1, & x\in[0,b),\\
\frac{\partial c}{\partial x}(t,t,x)<1, & x\in[b,\infty).
\end{cases}\label{eq:4-2-19}
\end{equation}

(i) Firstly, we show that the function $V_{4}$ defined by (\ref{eq:4-2-17})
is a concave function. Recalling Lemma \ref{lem:a-1} and $\lambda>0$,
we obtain 
\begin{align*}
V_{4}'(x) & =\frac{M}{\delta}\left(\frac{\lambda}{\mu-M-\sigma^{2}\theta_{3}}+\theta_{3}\left(1+\frac{\lambda}{\delta}\right)-\theta_{3}\frac{\lambda}{\mu-M-\sigma^{2}\theta_{3}}x\right)e^{-\theta_{3}x}\\
 & \geq\frac{M}{\delta}\left(\frac{\lambda}{\mu-M-\sigma^{2}\theta_{3}}+\theta_{3}\left(1+\frac{\lambda}{\delta}\right)\right)e^{-\theta_{3}x}\\
 & =\frac{M}{\delta}\left[\left(\frac{1}{\mu-M-\sigma^{2}\theta_{3}}+\frac{\theta_{3}}{\delta}\right)\lambda+\theta_{3}\right]e^{-\theta_{3}x}>0.
\end{align*}
Also note that $V_{3}(0)=V_{4}(0)=0$ and $V_{4}'(0)=\left[\frac{\lambda}{\mu-M-\sigma^{2}\theta_{3}}+\theta_{3}\left(1+\frac{\lambda}{\delta}\right)\right]\frac{M}{\delta}\in(0,1]$.
Recalling the second equation of (\ref{eq:4-2-4}), we have 
\begin{align*}
\frac{1}{2}\sigma^{2}V_{4}''(0) & =-\left(\mu-M\right)V_{4}'(0)+\delta V_{4}(0)-\lambda V_{3}(0)-M\\
 & =-\left(\mu-M\right)V_{4}'(0)-M\\
 & =-\mu V_{4}'(0)+M\left(V_{4}'(0)-1\right)\\
 & <0.
\end{align*}
Thus,
\begin{align*}
V_{4}''(x) & =-\theta_{3}\frac{M}{\delta}\left(\frac{2\lambda}{\mu-M-\sigma^{2}\theta_{3}}+\theta_{3}\left(1+\frac{\lambda}{\delta}\right)-\theta_{3}\frac{\lambda}{\mu-M-\sigma^{2}\theta_{3}}x\right)e^{-\theta_{3}x}\\
 & =\left[V_{4}''(0)+\theta_{3}^{2}\frac{M}{\delta}\frac{\lambda}{\mu-M-\sigma^{2}\theta_{3}}x\right]e^{-\theta_{3}x}<0.
\end{align*}
Therefore $\frac{\partial c}{\partial x}(t,t,x)=V_{4}'(x)\leq1$ for
all $x>0$. 

(ii) For $x\geq b$, recalling (\ref{eq:4-2-13}),
\begin{align*}
V_{4}'(x) & =\left(-\theta_{3}D_{3}+B_{3}-\theta_{3}B_{3}x\right)e^{-\theta_{3}x}\\
 & \geq\left(-\theta_{3}D_{3}+B_{3}-\theta_{3}B_{3}b\right)e^{-\theta_{3}x}\\
 & =-\left[\theta_{3}\left(D_{3}+B_{3}b\right)-B_{3}\right]e^{-\theta_{3}x}\\
 & =e^{\theta_{3}\left(b-x\right)}>0,
\end{align*}
and
\begin{eqnarray*}
V_{4}''(x) & = & \theta_{3}\left(\theta_{3}D_{3}-2B_{3}+\theta_{3}B_{3}x\right)e^{-\theta_{3}x}\\
 & \leq & \theta_{3}\left[\theta_{3}\left(D_{3}+B_{3}b\right)-2B_{3}\right]e^{-\theta_{3}x}\\
 & = & \theta_{3}\left(-B_{3}-e^{\theta_{3}b}\right)e^{-\theta_{3}x}\\
 & < & \theta_{3}\left(-\frac{\lambda}{\mu-M-\sigma^{2}\theta_{3}}\frac{M}{\delta}-1\right)e^{\theta_{3}\left(b-x\right)}\\
 & \leq & \theta_{3}\left(\theta_{3}\frac{M}{\delta}\frac{\lambda}{\delta}-1\right)e^{\theta_{3}\left(b-x\right)}.
\end{eqnarray*}
The last inequality follows from Lemma \ref{lem:a-1}. Furthermore,
by (\ref{eq:b-1}), we have $\theta_{3}\frac{M}{\delta}\frac{\lambda}{\delta}-1\leq0$.
Therefore, $V_{4}''(x)<0,$ for $x\geq b.$

Now we see the case when $0\leq x<b$. It follows from (\ref{eq:4-2-4})
and (\ref{eq:4-2-9}) that
\begin{align*}
\frac{1}{2}\sigma^{2}V_{4}''(b-) & =-\mu V_{4}'(b)+\delta V_{4}(b)-\lambda V_{3}(b),\\
\frac{1}{2}\sigma^{2}V_{4}''(b+) & =-\left(\mu-M\right)V_{4}'(b)+\delta V_{4}(b)-\lambda V_{3}(b)-M\\
 & =-\mu V_{4}'(b)+\delta V_{4}(b)-\lambda V_{3}(b),
\end{align*}
which yields that $V_{4}''(b+)=V_{4}''(b-)=V_{4}''(b).$ Furthermore,
for $0\leq x<b$,
\begin{align*}
V_{4}'''(x) & =\theta_{1}^{2}\left[\theta_{1}\hat{C}-3B_{1}-\theta_{1}B_{1}x\right]e^{\theta_{1}x}+\theta_{2}^{2}\left[\theta_{2}\hat{C}-3B_{1}+\theta_{2}B_{1}x\right]e^{-\theta_{2}x}\\
 & >\theta_{1}^{2}\left[\theta_{1}\hat{C}-3B_{1}-\theta_{1}B_{1}b\right]e^{\theta_{1}x}+\theta_{2}^{2}\left[\theta_{2}\hat{C}-3B_{1}\right]e^{-\theta_{2}x}.
\end{align*}
It follows from Lemma \ref{lem:4-2} that if (\ref{eq:b-1}) holds,
then $V_{4}'''(x)>0$ for $0\leq x<b$. Since $V_{4}''(x)$ is continuous
at $x=b$ and $V_{4}''(b)<0$, we get that $V_{4}''(x)<0,$ for $0\leq x<b$.
Therefore, $c(t,t,x)=V_{4}(x)$ is a concave function on $(0,\infty)$,
which together with (\ref{eq:4-2-9}) implies (\ref{eq:4-2-19}).
\end{proof}
Similar to the proof of Corollary \ref{cor:4-3}, it is easy to verify
(\ref{eq:3-4}). We have the following corollary immediately by Theorem
\ref{thm:3-2} and Theorem \ref{thm:4-6}.
\begin{cor}
Assume that $0<\lambda<\delta$. Consider the discount function (\ref{eq:4-2-1}).
\begin{enumerate}[label=(\roman{enumi})]
\item If $\frac{\delta}{M}>\theta_{3}$ and $\lambda\leq\left(\frac{\delta}{M}-\theta_{3}\right)\left[\frac{1}{\mu-M-\sigma^{2}\theta_{3}}+\frac{\theta_{3}}{\delta}\right]^{-1}$,
then for $(t,x)\in[0,\infty)\times[0,\infty)$, 
\[
\hat{\pi}(t,x)=\phi\left(t,t,\frac{\partial c}{\partial x}(t,t,x),\frac{\partial^{2}c}{\partial x^{2}}(t,t,x)\right)=M,
\]
 is an equilibrium dividend strategy, and 
\[
V(t,x)=c(t,t,x)=\left(1+\frac{\lambda}{\delta}\right)\frac{M}{\delta}+\frac{M}{\delta}\left[\frac{\lambda}{\mu-M-\sigma^{2}\theta_{3}}x-\left(1+\frac{\lambda}{\delta}\right)\right]e^{-\theta_{3}x},
\]
is the corresponding equilibrium value function.
\item If (\ref{eq:4-2-16}) and (\ref{eq:b-1}) hold, then for $t\in[0,\infty)$,
\[
\hat{\pi}(t,x)=\phi\left(t,t,\frac{\partial c}{\partial x}(t,t,x),\frac{\partial^{2}c}{\partial x^{2}}(t,t,x)\right)=\begin{cases}
0, & x\in[0,b),\\
M, & x\in[b,\infty),
\end{cases}
\]
is an equilibrium dividend strategy, and 
\[
V(t,x)=c(t,t,x)=\begin{cases}
\left(\hat{C}-B_{1}x\right)e^{\theta_{1}x}-\left(\hat{C}+B_{1}x\right)e^{-\theta_{2}x}, & x\in[0,b),\\
\left(1+\frac{\lambda}{\delta}\right)\frac{M}{\delta}+\left(D_{3}+B_{3}x\right)e^{-\theta_{3}x}, & x\in[b,\infty),
\end{cases}
\]
is the corresponding equilibrium value function. Here $\left(b,\hat{C},B_{1},B_{3},D_{3}\right)$
is the solution to (\ref{eq:4-2-9}).
\end{enumerate}
\end{cor}
\begin{example}
Let $\mu=1$, $\sigma=1$, $M=1$, $\delta=0.8$. Figure \ref{fig:4-2}
shows the equilibrium value functions for pseudo-exponential discount
functions with $\lambda=0,$ 0.1 and 0.2. The barriers $b$ are 0.3470,
0.4141 and 0.4796, respectively. The case with $\lambda=0$ is time
consistent and the equilibrium strategy is optimal. 
\end{example}
\begin{figure}
\begin{centering}
\includegraphics[scale=1.2]{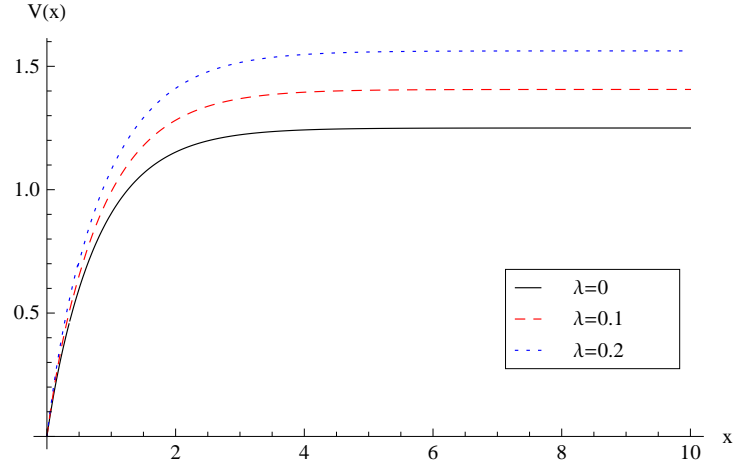}
\par\end{centering}

\caption{\label{fig:4-2}Equilibrium value functions with a pseudo-exponential
discount function }
\end{figure}

\appendix

\subsection*{Appendix A}

\setcounter{section}{1} \setcounter{equation}{0}\setcounter{thm}{0} 
\begin{lem}
\label{lem:a-1}$\frac{\theta_{1}}{\delta}-\frac{1}{\mu+\sigma^{2}\theta_{1}}>0$
and $\frac{1}{\mu-M-\sigma^{2}\theta_{3}}+\frac{\theta_{3}}{\delta}>0$.\end{lem}
\begin{proof}
Recall that $\theta_{1}$ and $\theta_{3}$ are given by 
\[
\theta_{1}=\frac{-\mu+\sqrt{\mu^{2}+2\sigma^{2}\delta}}{\sigma^{2}},\q\theta_{3}=\frac{\mu-M+\sqrt{\left(\mu-M\right)^{2}+2\sigma^{2}\delta}}{\sigma^{2}}.
\]
Then it follows that
\begin{align*}
\frac{\theta_{1}}{\delta}-\frac{1}{\mu+\sigma^{2}\theta_{1}}= & -\frac{1}{\sqrt{\mu^{2}+2\sigma^{2}\delta}}+\frac{-\mu+\sqrt{\mu^{2}+2\sigma^{2}\delta}}{\sigma^{2}\delta}\\
= & \frac{\left(-\mu+\sqrt{\mu^{2}+2\sigma^{2}\delta}\right)\sqrt{\mu^{2}+2\sigma^{2}\delta}-\sigma^{2}\delta}{\sigma^{2}\delta\sqrt{\mu^{2}+2\sigma^{2}\delta}}\\
= & \frac{-2\frac{\sqrt{2}}{2}\mu\sqrt{\frac{1}{2}\mu^{2}+\sigma^{2}\delta}+\left[\frac{1}{2}\mu^{2}+\sigma^{2}\delta\right]+\frac{1}{2}\mu^{2}}{\sigma^{2}\delta\sqrt{\mu^{2}+2\sigma^{2}\delta}}\\
= & \frac{\left[\frac{\sqrt{2}}{2}\mu-\sqrt{\frac{1}{2}\mu^{2}+\sigma^{2}\delta}\right]^{2}}{\sigma^{2}\delta\sqrt{\mu^{2}+2\sigma^{2}\delta}}>0,
\end{align*}
and
\begin{align*}
\frac{1}{\mu-M-\sigma^{2}\theta_{3}}+\frac{\theta_{3}}{\delta}= & -\frac{1}{\sqrt{\left(\mu-M\right)^{2}+2\sigma^{2}\delta}}+\frac{\mu-M+\sqrt{\left(\mu-M\right)^{2}+2\sigma^{2}\delta}}{\sigma^{2}\delta}\\
= & \frac{\left(\mu-M+\sqrt{\left(\mu-M\right)^{2}+2\sigma^{2}\delta}\right)\sqrt{\left(\mu-M\right)^{2}+2\sigma^{2}\delta}-\sigma^{2}\delta}{\sigma^{2}\delta\sqrt{\left(\mu-M\right)^{2}+2\sigma^{2}\delta}}\\
= & \frac{2\frac{\sqrt{2}}{2}\left(\mu-M\right)\sqrt{\frac{1}{2}\left(\mu-M\right)^{2}+\sigma^{2}\delta}+\left[\frac{1}{2}\left(\mu-M\right)^{2}+\sigma^{2}\delta\right]+\frac{1}{2}\left(\mu-M\right)^{2}}{\sigma^{2}\delta\sqrt{\left(\mu-M\right)^{2}+2\sigma^{2}\delta}}\\
= & \frac{\left[\frac{\sqrt{2}}{2}\left(\mu-M\right)+\sqrt{\frac{1}{2}\left(\mu-M\right)^{2}+\sigma^{2}\delta}\right]^{2}}{\sigma^{2}\delta\sqrt{\left(\mu-M\right)^{2}+2\sigma^{2}\delta}}>0.
\end{align*}

\end{proof}
\begin{proof}
[Proof of Lemma \ref{lem:4-5}]It is easy to check that 
\begin{eqnarray*}
G(0) & = & \left(\theta_{1}+\theta_{2}\right)\left\{ \lambda\left[\frac{1}{\mu-M-\sigma^{2}\theta_{3}}+\frac{\theta_{3}}{\delta}\right]\frac{M}{\delta}+\theta_{3}\frac{M}{\delta}-1\right\} 
\end{eqnarray*}
By Lemma \ref{lem:a-1} and (\ref{eq:4-2-16}) we have $G(0)>0$.
Now by (\ref{eq:4-2-0}), (\ref{eq:4-2-00}) and (\ref{eq:4-2-000}),
we rewrite $G(b)$ as 
\begin{eqnarray*}
G(b) & = & \theta_{3}\frac{\lambda}{\mu+\sigma^{2}\theta_{1}}\frac{M\theta_{3}}{\delta}\frac{1}{\left(\theta_{1}+\theta_{3}\right)e^{\theta_{1}b}+\left(\theta_{2}-\theta_{3}\right)e^{-\theta_{2}b}}\left(-e^{2\theta_{1}b}+e^{-2\theta_{2}b}\right)\\
 &  & +\frac{\lambda}{\mu-M-\sigma^{2}\theta_{3}}\frac{M}{\delta}\frac{1}{\left(\theta_{1}+\theta_{3}\right)e^{\theta_{1}b}+\left(\theta_{2}-\theta_{3}\right)e^{-\theta_{2}b}}\left(\theta_{1}e^{\theta_{1}b}+\theta_{2}e^{-\theta_{2}b}\right)^{2}\\
 &  & +2\left(\theta_{1}+\theta_{2}\right)\theta_{3}\frac{\lambda}{\mu+\sigma^{2}\theta_{1}}\frac{M\theta_{3}}{\delta}\frac{1}{\left(\theta_{1}+\theta_{3}\right)e^{\theta_{1}b}+\left(\theta_{2}-\theta_{3}\right)e^{-\theta_{2}b}}be^{\left(\theta_{1}-\theta_{2}\right)b}\\
 &  & +\left[\theta_{1}\theta_{3}\left(1+\frac{\lambda}{\delta}\right)\frac{M}{\delta}-\left(\theta_{1}+\theta_{3}\right)\right]e^{\theta_{1}b}+\left[\theta_{2}\theta_{3}\left(1+\frac{\lambda}{\delta}\right)\frac{M}{\delta}-\left(\theta_{2}-\theta_{3}\right)\right]e^{-\theta_{2}b}\\
 & := & \frac{1}{\left(\theta_{1}+\theta_{3}\right)e^{\theta_{1}b}+\left(\theta_{2}-\theta_{3}\right)e^{-\theta_{2}b}}g(b),
\end{eqnarray*}
where 
\begin{align*}
g(b):= & 2\left(\theta_{1}+\theta_{2}\right)\theta_{3}\frac{\lambda}{\mu+\sigma^{2}\theta_{1}}\frac{M\theta_{3}}{\delta}be^{\left(\theta_{1}-\theta_{2}\right)b}\\
 & +e^{2\theta_{1}b}\left\{ -\frac{\lambda}{\mu+\sigma^{2}\theta_{1}}\frac{M\theta_{3}^{2}}{\delta}+\frac{\lambda}{\mu-M-\sigma^{2}\theta_{3}}\frac{M\theta_{1}^{2}}{\delta}+\left[\theta_{1}\theta_{3}\left(1+\frac{\lambda}{\delta}\right)\frac{M}{\delta}-\left(\theta_{1}+\theta_{3}\right)\right]\left(\theta_{1}+\theta_{3}\right)\right\} \\
 & +e^{\left(\theta_{1}-\theta_{2}\right)b}\left\{ 2\theta_{1}\theta_{2}\frac{\lambda}{\mu-M-\sigma^{2}\theta_{3}}\frac{M}{\delta}+\left[\theta_{1}\theta_{3}\left(1+\frac{\lambda}{\delta}\right)\frac{M}{\delta}-\left(\theta_{1}+\theta_{3}\right)\right]\left(\theta_{2}-\theta_{3}\right)\right.\\
 & \left.+\left[\theta_{2}\theta_{3}\left(1+\frac{\lambda}{\delta}\right)\frac{M}{\delta}-\left(\theta_{2}-\theta_{3}\right)\right]\left(\theta_{1}+\theta_{3}\right)\right\} \\
 & +e^{-2\theta_{2}b}\left\{ \frac{\lambda}{\mu+\sigma^{2}\theta_{1}}\frac{M\theta_{3}^{2}}{\delta}+\frac{\lambda}{\mu-M-\sigma^{2}\theta_{3}}\frac{M\theta_{2}^{2}}{\delta}+\left[\theta_{2}\theta_{3}\left(1+\frac{\lambda}{\delta}\right)\frac{M}{\delta}-\left(\theta_{2}-\theta_{3}\right)\right]\left(\theta_{2}-\theta_{3}\right)\right\} \\
= & 2\left(\theta_{1}+\theta_{2}\right)\theta_{3}\frac{\lambda}{\mu+\sigma^{2}\theta_{1}}\frac{M\theta_{3}}{\delta}be^{\left(\theta_{1}-\theta_{2}\right)b}\\
 & +e^{2\theta_{1}b}\left\{ \lambda\frac{M}{\delta}\left[\theta_{3}^{2}\left(-\frac{1}{\mu+\sigma^{2}\theta_{1}}+\frac{\theta_{1}}{\delta}\right)+\theta_{1}^{2}\left(\frac{1}{\mu-M-\sigma^{2}\theta_{3}}+\frac{\theta_{3}}{\delta}\right)\right]+\theta_{1}\theta_{3}\left(\theta_{1}+\theta_{3}\right)\left(\frac{M}{\delta}-\frac{1}{\theta_{3}}-\frac{1}{\theta_{1}}\right)\right\} \\
 & +e^{\left(\theta_{1}-\theta_{2}\right)b}\left\{ \lambda\frac{M}{\delta}\left[2\theta_{1}\theta_{2}\frac{1}{\mu-M-\sigma^{2}\theta_{3}}+\theta_{3}\left(2\theta_{1}\theta_{2}-\theta_{1}\theta_{3}+\theta_{2}\theta_{3}\right)\frac{1}{\delta}\right]\right.\\
 & \left.+\theta_{3}\frac{M}{\delta}\left(2\theta_{1}\theta_{2}-\theta_{1}\theta_{3}+\theta_{2}\theta_{3}\right)-2\left(\theta_{2}-\theta_{3}\right)\left(\theta_{1}+\theta_{3}\right)\right\} \\
 & +e^{-2\theta_{2}b}\left\{ \lambda\frac{M}{\delta}\left[\theta_{3}^{2}\left(\frac{1}{\mu+\sigma^{2}\theta_{1}}-\frac{\theta_{2}}{\delta}\right)+\theta_{2}^{2}\left(\frac{1}{\mu-M-\sigma^{2}\theta_{3}}+\frac{\theta_{3}}{\delta}\right)\right]+\left(\theta_{2}-\theta_{3}\right)\left[\theta_{2}\theta_{3}\frac{M}{\delta}-\left(\theta_{2}-\theta_{3}\right)\right]\right\} .
\end{align*}
From Lemma 2.1 of \citet{at97}, we have
\[
\frac{M}{\delta}-\frac{1}{\theta_{3}}-\frac{1}{\theta_{1}}<0.
\]
By Lemma \ref{lem:a-1} and (\ref{eq:4-2-16}), it is easy to see
that $G(\infty)<0$. Thus, the equation $G(b)=0$ admits a positive
solution. 
\end{proof}

\section*{Appendix B}

\setcounter{section}{2} \setcounter{equation}{0} 
\begin{lem}
\label{lem:4-2}If 
\begin{equation}
\lambda\leq\frac{\theta_{1}+\theta_{2}}{\theta_{1}+3\theta_{2}}\frac{\delta^{2}}{M\theta_{3}}\wedge\frac{\left(\theta_{1}+\theta_{3}\right)\left(\theta_{1}+\theta_{2}\right)}{2\theta_{1}\left(\theta_{1}+2\theta_{2}\right)}\frac{\delta^{2}}{M\theta_{3}},\label{eq:b-1}
\end{equation}
then
\[
\theta_{1}\hat{C}-3B_{1}-\theta_{1}B_{1}b>0.
\]
\end{lem}
\begin{proof}
It follows that
\begin{align}
 & \theta_{1}\hat{C}-3B_{1}-\theta_{1}B_{1}b\nonumber \\
= & \frac{\theta_{1}-2\theta_{1}\theta_{2}bB_{1}e^{-\theta_{2}b}-2B_{1}\theta_{1}e^{\theta_{1}b}+B_{1}\left(\theta_{1}-3\theta_{2}\right)e^{-\theta_{2}b}}{\theta_{1}e^{\theta_{1}b}+\theta_{2}e^{-\theta_{2}b}}\nonumber \\
= & \frac{\theta_{1}+B_{1}\left[-2\theta_{1}e^{\theta_{1}b}+\left(-2\theta_{1}\theta_{2}b+\theta_{1}-3\theta_{2}\right)e^{-\theta_{2}b}\right]}{\theta_{1}e^{\theta_{1}b}+\theta_{2}e^{-\theta_{2}b}}\nonumber \\
= & \frac{\theta_{1}\left[\theta_{1}+\theta_{3}-2\frac{\lambda}{\mu+\sigma^{2}\theta_{1}}\frac{M\theta_{3}}{\delta}\right]e^{\left(\theta_{1}+\theta_{2}\right)b}+\left[\theta_{1}\left(\theta_{2}-\theta_{3}\right)+\frac{\lambda}{\mu+\sigma^{2}\theta_{1}}\frac{M\theta_{3}}{\delta}\left(-2\theta_{1}\theta_{2}b+\theta_{1}-3\theta_{2}\right)\right]}{e^{\theta_{2}b}\left(\theta_{1}e^{\theta_{1}b}+\theta_{2}e^{-\theta_{2}b}\right)\left[\left(\theta_{1}+\theta_{3}\right)e^{\theta_{1}b}+\left(\theta_{2}-\theta_{3}\right)e^{-\theta_{2}b}\right]}.\label{eq:b-2}
\end{align}
Let 
\[
q(b):=\theta_{1}\left[\theta_{1}+\theta_{3}-2\frac{\lambda}{\mu+\sigma^{2}\theta_{1}}\frac{M\theta_{3}}{\delta}\right]e^{\left(\theta_{1}+\theta_{2}\right)b}+\left[\theta_{1}\left(\theta_{2}-\theta_{3}\right)+\frac{\lambda}{\mu+\sigma^{2}\theta_{1}}\frac{M\theta_{3}}{\delta}\left(-2\theta_{1}\theta_{2}b+\theta_{1}-3\theta_{2}\right)\right].
\]
Then
\begin{eqnarray*}
q'(b) & = & \theta_{1}\left(\theta_{1}+\theta_{2}\right)\left[\theta_{1}+\theta_{3}-2\frac{\lambda}{\mu+\sigma^{2}\theta_{1}}\frac{M\theta_{3}}{\delta}\right]e^{\left(\theta_{1}+\theta_{2}\right)b}-2\theta_{1}\theta_{2}\frac{\lambda}{\mu+\sigma^{2}\theta_{1}}\frac{M\theta_{3}}{\delta},\\
q''(b) & = & \theta_{1}\left(\theta_{1}+\theta_{2}\right)^{2}\left[\theta_{1}+\theta_{3}-2\frac{\lambda}{\mu+\sigma^{2}\theta_{1}}\frac{M\theta_{3}}{\delta}\right]e^{\left(\theta_{1}+\theta_{2}\right)b},
\end{eqnarray*}
and 
\begin{eqnarray*}
q(0) & = & \theta_{1}\left[\theta_{1}+\theta_{3}-2\frac{\lambda}{\mu+\sigma^{2}\theta_{1}}\frac{M\theta_{3}}{\delta}\right]+\theta_{1}\left(\theta_{2}-\theta_{3}\right)+\frac{\lambda}{\mu+\sigma^{2}\theta_{1}}\frac{M\theta_{3}}{\delta}\left(\theta_{1}-3\theta_{2}\right)\\
 & = & \theta_{1}\left(\theta_{1}+\theta_{2}\right)-\left(\theta_{1}+3\theta_{2}\right)\frac{\lambda}{\mu+\sigma^{2}\theta_{1}}\frac{M\theta_{3}}{\delta},\\
q'(0) & = & \theta_{1}\left(\theta_{1}+\theta_{2}\right)\left[\theta_{1}+\theta_{3}-2\frac{\lambda}{\mu+\sigma^{2}\theta_{1}}\frac{M\theta_{3}}{\delta}\right]-2\frac{\lambda}{\mu+\sigma^{2}\theta_{1}}\frac{M\theta_{3}}{\delta}\theta_{1}\theta_{2}\\
 & = & \theta_{1}\left(\theta_{1}+\theta_{2}\right)\left(\theta_{1}+\theta_{3}\right)-2\theta_{1}\left(\theta_{1}+2\theta_{2}\right)\frac{\lambda}{\mu+\sigma^{2}\theta_{1}}\frac{M\theta_{3}}{\delta}.
\end{eqnarray*}
If $\lambda\leq\frac{\theta_{1}+\theta_{2}}{\theta_{1}+3\theta_{2}}\frac{\delta^{2}}{M\theta_{3}}\wedge\frac{\left(\theta_{1}+\theta_{3}\right)\left(\theta_{1}+\theta_{2}\right)}{2\theta_{1}\left(\theta_{1}+2\theta_{2}\right)}\frac{\delta^{2}}{M\theta_{3}}$
holds, then it follows from Lemma \ref{lem:a-1} that 
\[
q(0)=\left(\theta_{1}+\theta_{2}\right)\left[\theta_{1}-\frac{\theta_{1}+3\theta_{2}}{\theta_{1}+\theta_{2}}\frac{\lambda}{\mu+\sigma^{2}\theta_{1}}\frac{M\theta_{3}}{\delta}\right]\geq\delta\left(\theta_{1}+\theta_{2}\right)\left(\frac{\theta_{1}}{\delta}-\frac{1}{\mu+\sigma^{2}\delta}\right)>0.
\]
and similarly, 
\[
q'(0)=\left(\theta_{1}+\theta_{2}\right)\left(\theta_{1}+\theta_{3}\right)\left[\theta_{1}-\frac{2\theta_{1}\left(\theta_{1}+2\theta_{2}\right)}{\left(\theta_{1}+\theta_{2}\right)\left(\theta_{1}+\theta_{3}\right)}\frac{\lambda}{\mu+\sigma^{2}\theta_{1}}\frac{M\theta_{3}}{\delta}\right]\geq\delta\left(\theta_{1}+\theta_{2}\right)\left(\theta_{1}+\theta_{3}\right)\left(\frac{\theta_{1}}{\delta}-\frac{1}{\mu+\sigma^{2}\delta}\right)>0.
\]
Thus, it follows from $q'(0)\geq0$ that 
\[
\theta_{1}\left(\theta_{1}+\theta_{2}\right)\left[\theta_{1}+\theta_{3}-2\frac{\lambda}{\mu+\sigma^{2}\theta_{1}}\frac{M\theta_{3}}{\delta}\right]\geq2\frac{\lambda}{\mu+\sigma^{2}\theta_{1}}\frac{M\theta_{3}}{\delta}\theta_{1}\theta_{2}>0.
\]
Therefore, $q''(b)>0,$ $q'(b)>0$, and then $q(b)>0$. Finally, it
follows from (\ref{eq:b-2}) that 
\[
\theta_{1}\hat{C}-3B_{1}-\theta_{1}B_{1}b>0.
\]

\end{proof}

\section*{Acknowledgments}

We would like to thank the referee(s) for valuable comments and suggestions.
This work was supported by National Natural Science Foundation of
China (10971068, 11231005), Doctoral Program Foundation of the Ministry
of Education of China (20110076110004), Program for New Century Excellent
Talents in University (NCET-09-0356) and \textquotedblleft{}the Fundamental
Research Funds for the Central Universities\textquotedblright{}. 

\bibliographystyle{plainnat}
\bibliography{mybib}

\end{document}